\documentclass[letterpaper, 10 pt, conference]{ieeeconf}%

\IEEEoverridecommandlockouts
\usepackage{cite}
\usepackage{amsmath,amssymb,amsfonts}
\usepackage{algorithmic}
\usepackage{graphicx}
\usepackage{textcomp}
\usepackage{xcolor}
\usepackage{cuted}
\usepackage{lipsum}
\usepackage{mathtools}
\usepackage{stfloats}
\usepackage{bbm}
\usepackage{multirow}

\newcommand{\R}{\mathbb{R}}
\newcommand{\N}{\mathbb{N}}
\newcommand{\T}{\top}
\newcommand{\I}{\mathbf{I}}
\newcommand{\0}{\mathbf{0}}
\newcommand{\E}{\mathcal{E}}

\newcommand{\diag}{\text{diag}}

\newcommand{\bm}[1]{\begin{bmatrix}#1\end{bmatrix}}
\def\BibTeX{{\rm B\kern-.05em{\sc i\kern-.025em b}\kern-.08em
    T\kern-.1667em\lower.7ex\hbox{E}\kern-.125emX}}

\newtheorem{assumption}{\textbf{Assumption}}
\newtheorem{theorem}{\textbf{Theorem}}

\newtheorem{lemma}{\textbf{Lemma}}
\newtheorem{proposition}{\textbf{Proposition}}
\newtheorem{remark}{\textbf{Remark}}
\newtheorem{definition}{\textbf{Definition}}

\usepackage{hyperref}

\begin{document}

\title{
{\LARGE \textbf{
Dissipativity-Based Synthesis of Distributed Control and Communication Topology Co-Design for AC Microgrids
}}}



\author{Mohammad Javad Najafirad, Shirantha Welikala, Lei Wu, and Panos J. Antsaklis%
\thanks{Mohammad Javad Najafirad, Shirantha Welikala, and Lei Wu are with the Department of Electrical and Computer Engineering, School of Engineering and Science, Stevens Institute of Technology, Hoboken, NJ 07030, \texttt{\small\{mnajafir, swelikal, lei.wu\}@stevens.edu}. 
Panos J. Antsaklis is with the Department of Electrical Engineering, University of Notre Dame, Notre Dame, IN 46556, \texttt{\small pantsakl@nd.edu}.}}

\maketitle


\begin{abstract}
This paper introduces a dissipativity-based framework 
for the joint design of distributed controllers and communication topologies in AC microgrids (MGs), providing robust performance guarantees for voltage regulation, frequency synchronization, and proportional power sharing across distributed generators (DGs). The closed-loop AC MG is represented as a networked system in which DGs, distribution lines, and loads function as interconnected subsystems linked through cyber-physical networks. Each DG 
utilizes a three-layer hierarchical control structure: a steady-state controller for operating point configuration, a local feedback controller for voltage tracking, and a distributed droop-free controller implementing normalized power consensus for frequency coordination and proportional power distribution. The operating point design is formulated as an optimization problem. Leveraging dissipativity theory, we derive necessary and sufficient subsystem dissipativity conditions. The global co-design is then cast as a convex linear matrix inequality (LMI) optimization that jointly determines distributed controller parameters and sparse communication 
architecture while managing the highly nonlinear, coupled dq-frame dynamics characteristic of AC systems. Simulation results from an islanded AC MG in a MATLAB/Simulink environment verify that the proposed framework achieves robust voltage regulation, frequency synchronization, and proportional power sharing through the optimized communication topology.
\end{abstract}


\section{Introduction} 

The integration of renewable energy sources into power systems has accelerated the development of microgrids (MGs) as a reliable solution for electric energy sustainability. AC MGs, in particular, have emerged as critical components of modern power infrastructure and provide enhanced flexibility to manage distributed generators (DGs), energy storage systems, and various load types within localized networks \cite{han2016review}. Nevertheless, the control architecture of AC MGs faces significant challenges to maintain voltage stability, frequency synchronization, and proportional power sharing among DGs, particularly in islanded operation where no utility grid support is available \cite{olivares2014trends}. These challenges become more complex due to the inherent nature of AC systems, which must manage both active and reactive power flows while multiple DGs operate in parallel synchronization \cite{rocabert2012control}. Recent advances in networked control theory and communication technologies have opened new avenues to develop sophisticated control strategies that address these multifaceted challenges through coordinated operation rather than purely local control actions \cite{sahoo2019cyber}.

Traditional control strategies for AC MGs have evolved from centralized to decentralized and distributed approaches to address scalability and reliability concerns \cite{bidram2012hierarchical}. Centralized control offers precise coordination but suffers from single points of failure and communication bottlenecks, while fully decentralized methods sacrifice optimal performance for local autonomy \cite{shafiee2013distributed}. Distributed control emerges as a balanced solution that enables DGs to coordinate through sparse communication networks while each unit maintains local control authority \cite{guo2014distributed}. The challenge lies in the design of both the control algorithms and the communication topology to achieve desired performance objectives. Most existing distributed control methods assume predetermined communication structures that mirror the physical electrical connections, which may not offer optimal information exchange patterns \cite{zuo2016distributed}. Furthermore, the complex nature of AC systems requires careful consideration of both voltage magnitude and phase angle dynamics, as well as the coupled relationship between active and reactive power flows \cite{simpson2016voltage}.

Dissipativity theory provides a powerful framework to analyze and design control systems for interconnected networks by characterizing subsystems through their energy exchange properties \cite{arcak2016networks}. Unlike traditional Lyapunov methods that require detailed system models, dissipativity-based approaches need only input-output energy relationships, greatly simplifying the analysis of complex networked systems \cite{de2017bregman}. Recent applications to AC MGs have shown promise in addressing voltage stabilization and power quality issues through passivity-based methods, though most approaches still rely on predetermined control structures without considering communication optimization \cite{kwasinski2007stabilization}. The compositional nature of dissipativity theory allows modular design where individual DG controllers can be synthesized independently while guaranteeing overall system stability, a critical feature for scalable AC MG architectures \cite{simpson2016voltage}. However, the extension to comprehensive AC system control presents unique challenges due to the coupling between voltage components in the $dq$ reference frame, the inherent complexity of three-phase power systems, and the need to manage both active and reactive power flows \cite{schiffer2016survey}. The ability to guarantee stability through compositional certification makes dissipativity particularly attractive for modular AC MG architectures.

Despite these theoretical advances, most existing distributed control methods for AC MGs treat controller design and communication topology as separate problems, potentially missing optimal solutions that arise from their joint consideration \cite{trip2017distributed}. Current approaches typically assume fixed communication structures or require all-to-all information exchange, which becomes impractical as MG size increases \cite{meng2017review}. Moreover, simultaneously achieving voltage regulation and proportional power sharing in AC systems requires careful coordination of both active and reactive power controllers, a challenge not fully addressed by existing methods \cite{khayat2019secondary}. The need for a systematic co-design framework that jointly optimizes control parameters and communication topology while guaranteeing stability through rigorous theoretical foundations motivates this work.

This paper presents a novel dissipativity-based control and communication topology co-design framework for AC MGs with the following contributions:

\begin{enumerate}
    \item We formulate the AC MG as a networked system, where DGs, lines, and loads are modeled as dynamic subsystems, while cyber and physical links are modeled using static interconnection matrices.
    \item We propose an efficient LMI-based framework for co-designing distributed controllers and communication topologies, enabling the trade-off between control performance and communication costs. 
    \item We derive the AC MG operating (equilibrium) point design problem, as well as the subsystem-level dissipativity analysis and design problems, as convex LMI problems to support the global co-design framework. 
    \item The proposed comprehensive AC MG control framework enables the enforcement of robustness guarantees for AC MG voltage regulation, frequency synchronization, and proportional power sharing performance.
\end{enumerate}

This work builds upon the dissipativity-based framework established in \cite{najafirad2025dissipativity} for DC MGs and the systematic co-design methodology developed in \cite{welikala2025decentralized} for vehicular platoons, extending these concepts to address the unique challenges of AC MG control with coupled $dq$ dynamics and dual power flow objectives.


\section{Preliminaries}\label{Preliminaries}

\subsubsection*{\textbf{Notations}}
We use $\mathbb{R}$ and $\mathbb{N}$ to denote the sets of real and natural numbers, respectively. 
For any $N\in\mathbb{N}$, we define $\mathbb{N}_N\triangleq\{1,2,..,N\}$.
A block matrix $A$ of dimension $n \times m$ is expressed as $A = [A_{ij}]_{i \in \mathbb{N}_n, j \in \mathbb{N}_m}$. $\0$ and $\I$ are the zero and identity matrices, respectively. For a symmetric matrix $A\in\mathbb{R}^{n\times n}$ we write $A>0\ (A\geq0)$ to indicate positive definiteness (semi-definiteness). $\star$ represents conjugate blocks in symmetric matrices. We use $\mathcal{H}(A)\triangleq A + A^\T$ to represent the symmetric part of a matrix.

\subsection{Dissipativity}
We consider a nonlinear dynamic system:
\begin{equation}\label{dynamic}
\begin{aligned}
    \dot{x}(t)=f(x(t),u(t)),\quad
    y(t)=h(x(t),u(t)),
    \end{aligned}
\end{equation}
where $x(t)\in\mathbb{R}^n$ denotes the state, $u(t)\in\mathbb{R}^q$ the control input, $y(t)\in\mathbb{R}^m$ the output, and $f:\mathbb{R}^n\times\mathbb{R}^q\rightarrow\mathbb{R}^n$ and $h:\mathbb{R}^n\times\mathbb{R}^q\rightarrow\mathbb{R}^m$ are continuously differentiable functions, satisfying $f(\0,\0)=\0$ and $h(\0,\0)=\0$.

\begin{definition}\cite{arcak2022}
The system \eqref{dynamic} is said to be dissipative with respect to the supply rate $s:\mathbb{R}^q\times\mathbb{R}^m\rightarrow\mathbb{R}$ if there exists a continuously differentiable storage function $V:\mathbb{R}^n\times\mathcal{X}\rightarrow\mathbb{R}$ satisfying  $V(x)>0, \forall  x \neq \0$, $V(\0)=\0$, and $\dot{V}(x)=\nabla_x V(x)f(x,u) \leq s(u,y)$, for all $(x,u)\in\mathbb{R}^n\times\mathbb{R}^q$.
\end{definition}


\begin{definition} \cite{welikala2023non}
The system (\ref{dynamic}) is $X$-dissipative (where $(X\triangleq [X^{kl}]_{k,l\in\N_2} = X^\T)$) when it is dissipative under the supply rate:
\begin{center}
$
s(u,y)\triangleq
\begin{bmatrix}
    u \\ y
\end{bmatrix}^\top
\begin{bmatrix}
    X^{11} & X^{12}\\ X^{21} & X^{22}
\end{bmatrix}
\begin{bmatrix}
    u \\ y
\end{bmatrix}.
$
\normalsize
\end{center}
\end{definition}

\begin{remark}\label{Rm:X-DissipativityVersions}\cite{welikala2022line}
The system (\ref{dynamic}) with $X$-dissipativity has the following properties depending on the structure of $X$:\\
1)\ When $X = \scriptsize\begin{bmatrix}
    \0 & \frac{1}{2}\I \\ \frac{1}{2}\I & \0
\end{bmatrix}\normalsize$, the system is passive;\\
2)\ When $X = \scriptsize\begin{bmatrix}
    -\nu\I & \frac{1}{2}\I \\ \frac{1}{2}\I & -\rho\I
\end{bmatrix}\normalsize$, the system is strictly passive with input/output passivity indices $\nu$ and $\rho$, denoted as IF-OFP($\nu,\rho$));\\
3)\ When $X = \scriptsize\begin{bmatrix}
    \gamma^2\I & \0 \\ \0 & -\I
\end{bmatrix}\normalsize$, the system is $L_2$-stable with gain $\gamma$, denoted as $L2G(\gamma)$). 
\end{remark}

For a linear time-invariant (LTI) system \eqref{dynamic}, $X$-dissipativity can be verified through the following LMI condition.

\begin{proposition}\label{Prop:linear_X-EID} \cite{welikala2025decentralized}
An LTI system described by
\begin{equation*}\label{Eq:Prop:linear_X-EID_1}
\begin{aligned}
    \dot{x}(t)=Ax(t)+Bu(t),\quad 
    y(t)=Cx(t)+Du(t),
\end{aligned}
\end{equation*}
is $X$-dissipative if and only if there exists $P>0$ such that
\begin{equation*}\label{Eq:Prop:linear_X-EID_2}
\scriptsize
\begin{bmatrix}
-\mathcal{H}(PA)+C^\top X^{22}C & -PB+C^\top X^{21}+C^\top X^{22}D\\
\star & X^{11}+\mathcal{H}(X^{12}D)+D^\top X^{22}D
\end{bmatrix}
\normalsize
\geq0.
\end{equation*}
\end{proposition}


\subsection{Networked Systems} \label{SubSec:NetworkedSystemsPreliminaries}
We examine a networked system $\Sigma$ composed of dynamic subsystems $\Sigma_i,i\in\mathbb{N}_N$,  $\Bar{\Sigma}_i,i\in\mathbb{N}_{\Bar{N}}$, and $\check{\Sigma}_m,m\in\mathbb{N}_M$, which are coupled via a static interconnection matrix $M$, with the exogenous inputs $w(t)\in\mathbb{R}^r$ (e.g., disturbances) and interested outputs $z(t)\in\mathbb{R}^l$ (e.g., performance). 


Each subsystem $\Sigma_i,i\in\mathbb{N}_N$ has dynamics
\begin{equation*}
    \begin{aligned}
        \dot{x}_i(t)=f_i(x_i(t),u_i(t)),\quad
        y_i(t)=h_i(x_i(t),u_i(t)),
    \end{aligned}
\end{equation*}
where $x_i(t)\in\mathbb{R}^{n_i}$ represents the state, $u_i(t)\in\mathbb{R}^{q_i}$ the input, and $y_i(t)\in\mathbb{R}^{m_i}$ output. We assume each $\Sigma_i$ is $X_i$-dissipative with $X_i \triangleq [X_i^{kl}]_{k,l\in\N_2}$. For subsystems $\bar{\Sigma}_l, l\in\N_{\bar{N}}$ and $\check{\Sigma}_m, m\in\N_{\check{N}}$, we employ similar notation with bar and check symbols to indicate $\bar{X}_l$-dissipative and $\check{X}_m$-dissipative where $\bar{X}_l \triangleq [\bar{X}_l^{kl}]_{k,l\in\N_2}$ and $\check{X}_m \triangleq [\check{X}_m^{kl}]_{k,l\in\N_2}$.

By defining $u\triangleq[u_i^\top]^\top_{i\in\mathbb{N}_N}$, $y\triangleq[y_i^\top]^\top_{i\in\mathbb{N}_N}$, $\Bar{u}\triangleq[\Bar{u}_i^\top]^\top_{i\in\mathbb{N}_{\Bar{N}}}$, $\Bar{y}\triangleq[y_i^\top]^\top_{i\in\mathbb{N}_{\Bar{N}}}$, $\check{u}\triangleq[\check{u}_m^\top]^\top_{m\in\mathbb{N}_{\check{N}}}$ and $\check{y}\triangleq[\check{y}_m^\top]^\top_{m\in\mathbb{N}_{\check{N}}}$, the interconnection relationships are 
expressed through matrix $M$ as:
\begin{equation*}\label{interconnectionMatrix}
\scriptsize
\begin{bmatrix}
    u \\ \bar{u} \\ \check{u} \\ z
\end{bmatrix}=M
\normalsize
\scriptsize
\begin{bmatrix}
    y \\ \bar{y} \\ \check{y} \\ w
\end{bmatrix}
\normalsize
\equiv
\scriptsize
\begin{bmatrix}
    M_{uy} & M_{u\bar{y}} & M_{u\check{y}} & M_{uw} \\
    M_{\bar{u}y} & M_{\bar{u}\bar{y}} & M_{\bar{u}\check{y}} & M_{\bar{u}w} \\
    M_{\check{u}y} & M_{\check{u}\bar{y}} & M_{\check{u}\check{y}} & M_{\check{u}w} \\
    M_{zy} & M_{z\bar{y}} & M_{z\check{y}} & M_{zw}
\end{bmatrix} \begin{bmatrix} y \\ \bar{y} \\ \check{y} \\ w
\end{bmatrix}.
\normalsize
\end{equation*}

The following result leverages the dissipativity of all three subsystem types to establish an LMI-based problem for synthesizing the interconnection matrix $M$, which guarantees $\textbf{Y}$-dissipativity of the overall networked system $\Sigma$ under two mild conditions \cite{welikala2023non}.

\begin{assumption}\label{As:NegativeDissipativity}
    For the networked system $\Sigma$, the provided \textbf{Y}-dissipative specification is such that $\textbf{Y}^{22}<0$.
\end{assumption}

\begin{remark}
According to Remark \ref{Rm:X-DissipativityVersions}, Assumption \ref{As:NegativeDissipativity} is satisfied when the networked system $\Sigma$ is required to be either: (i) L2G($\gamma$) or (ii) IF-OFP($\nu,\rho$) with some $\rho>0$, i.e., $L_2$-stable or passive, respectively. This Assumption \ref{As:NegativeDissipativity} is mild since ensuring either $L_2$-stability or passivity is typically desired for networked systems.
\end{remark}

\begin{assumption}\label{As:PositiveDissipativity}
    Within the networked system $\Sigma$, each subsystem $\Sigma_i$ is $X_i$-dissipative with $X_i^{11}>0, \forall i\in\mathbb{N}_N$. Similarly, each subsystem $\bar{\Sigma}_i$ is $\bar{X}_i$-dissipative with $\bar{X}_i^{11}>0, \forall i\in\N_{\bar{N}}$, and each subsystem $\check{\Sigma}_m$ is $\check{X}_m$-dissipative with $\check{X}_m^{11}>0, \forall m\in\N_{\check{N}}$.
\end{assumption}

\begin{remark}
    Based on Remark \ref{Rm:X-DissipativityVersions}, Assumption \ref{As:PositiveDissipativity} holds when a subsystem $\Sigma_i,i\in\N_N$ is either: (i) L2G($\gamma_i$) or (ii) IF-OFP($\nu_i,\rho_i$) with $\nu_i<0$. Since passivity-based control frequently involves non-passive or strictly input passive subsystems, this assumption \ref{As:PositiveDissipativity} is also mild in practical scenarios.
\end{remark}

\begin{proposition}\label{synthesizeM}\cite{welikala2023non}
    Under Assumptions \ref{As:NegativeDissipativity} and \ref{As:PositiveDissipativity}, the network system $\Sigma$ can be made \textbf{Y}-dissipative (from $w(t)$ to $z(t)$) by synthesizing the interconnection matrix $M$ through solving the LMI problem:
\begin{equation*}
\begin{aligned}
	&\mbox{Find: } 
	L_{uy}, L_{u\bar{y}}, L_{u\check{y}}, L_{uw}, L_{\bar{u}y}, L_{\bar{u}\bar{y}}, L_{\bar{u}\check{y}}, L_{\bar{u}w}, L_{\check{u}y}, L_{\check{u}\bar{y}},\\ &L_{\check{u}\check{y}}, L_{\check{u} w}, M_{zy}, M_{z\bar{y}}, M_{z\check{y}}, \text{and} \ M_{zw}, \\
	&\mbox{s.t.: } \hspace{2mm} p_i \geq 0, \forall i\in\N_N, 
	\bar{p}_l \geq 0, \forall l\in\N_{\bar{N}},\check{p}_m \geq 0, \forall m\in\N_{\check{N}},\ \\ &\text{and} \ \eqref{NSC4YEID}, \text{with}\\
&\left[\begin{smallmatrix}
M_{uy} & M_{u\bar{y}} & M_{u\check{y}}& M_{uw} \\ M_{\bar{u}y} & M_{\bar{u}\bar{y}} &  M_{\bar{u}\check{y}} & M_{\bar{u}w} \\ M_{\check{u}y} & M_{\check{u}\bar{y}} & M_{\check{u}\check{y}} & M_{\check{u}w}
\end{smallmatrix}\right]\negthickspace=\negthickspace
\left[\begin{smallmatrix}
\textbf{X}_p^{11} & \0 & \0 \\ \0 & \bar{\textbf{X}}_{\bar{p}}^{11} & \0 \\ \0 & \0 & \check{\textbf{X}}_{\check{p}}^{11}
\end{smallmatrix}\right]^{-1} 
\left[\begin{smallmatrix}
L_{uy} & L_{u\bar{y}} & L_{u\check{y}} & L_{uw} \\ L_{\bar{u}y} & L_{\bar{u}\bar{y}} & L_{\bar{u}\check{y}} & L_{\bar{u}w} \\  L_{\check{u}y} & L_{\check{u}\bar{y}} & L_{\check{u}\check{y}} & L_{\check{u}w}
\end{smallmatrix}\right],
\end{aligned}
\end{equation*}
where $\textbf{X}_p^{11} = \diag(\{p_iX_i^{11}:i\in\N_N\})$, $\bar{\textbf{X}}_{\bar{p}}^{11} = \diag(\{\bar{p}_i\bar{X}_i^{11}:i\in\N_{\bar{N}}\})$, $\check{\textbf{X}}_{\check{p}}^{11} = \diag(\{\check{p}_m\check{X}_m^{11}:m\in\N_{\check{N}}\})$, $\textbf{X}^{12} \triangleq \diag((X_i^{11})^{-1}X_i^{12}:i\in\N_N)$, 
$\textbf{X}^{21} \triangleq (\textbf{X}^{12})^\T$, and similarly for $\bar{\mathbf{X}}^{12}$, $\bar{\mathbf{X}}^{21}$, $\check{\mathbf{X}}^{12}$, $\check{\mathbf{X}}^{21}$.
\end{proposition}

\begin{figure*}[!hb]
\vspace{-5mm}
\centering
\hrulefill
\begin{equation}\label{NSC4YEID}
\left[\begin{smallmatrix}
\textbf{X}_p^{11} & \0 & \0 & \0 & L_{uy} & L_{u\bar{y}} & L_{u\check{y}} & L_{uw} \\
\0 & \bar{\textbf{X}}_{\bar{p}}^{11} & \0 & \0 & L_{\bar{u}y} & L_{\bar{u}\bar{y}} & L_{\bar{u}\check{y}} & L_{\bar{u}w} \\
\0 & \0 & \check{\textbf{X}}_{\check{p}}^{11} & \0 & L_{\check{u}y} & L_{\check{u}\bar{y}} & L_{\check{u}\check{y}} & L_{\check{u}w} \\
\0 & \0 & \0 & -\mathbf{Y}^{22} & -\mathbf{Y}^{22} M_{zy} & -\mathbf{Y}^{22} M_{z\bar{y}} & -\mathbf{Y}^{22} M_{z\check{y}} & -\mathbf{Y}^{22} M_{zw} \\
L_{uy}^\T & L_{\bar{u}y}^\T & L_{\check{u}y}^\T & -M_{zy}^\T\mathbf{Y}^{22} & -L_{uy}^\T \textbf{X}^{12} - \textbf{X}^{21}L_{uy} - \textbf{X}_p^{22} & -\textbf{X}^{21}L_{u\bar{y}} - L_{uy}^\T\bar{\textbf{X}}^{12} & -\textbf{X}^{21}L_{u\check{y}} - L_{uy}^\T\check{\textbf{X}}^{12} & -\textbf{X}^{21}L_{uw} + M_{zy}^\T \textbf{Y}^{21} \\
L_{u\bar{y}}^{\top} & L_{\bar{u}\bar{y}}^\T & L_{\check{u}\bar{y}}^\T & -M_{z\bar{y}}^\T\mathbf{Y}^{22} &-L_{u\bar{y}}^\T \textbf{X}^{12} - \bar{\textbf{X}}^{21}L_{uy} & -L_{\bar{u}\bar{y}}^\T \bar{\textbf{X}}^{12} - \bar{\textbf{X}}^{21}L_{\bar{u}\bar{y}} - \bar{\textbf{X}}_{\bar{p}}^{22} & -\bar{\textbf{X}}^{21} L_{\bar{u}\check{y}} - L_{\bar{u}\bar{y}}^\T \check{\textbf{X}}^{12} & -\bar{\textbf{X}}^{21}L_{\bar{u}w} + M_{z\bar{y}}^{\top}\textbf{Y}^{21} \\
L_{u\check{y}}^\T & L_{\bar{u}\check{y}}^\T & L_{\check{u}\check{y}}^\T & -M_{z\check{y}}^\T\mathbf{Y}^{22} &-L_{u\check{y}}^\T \textbf{X}^{12} - \check{\textbf{X}}^{21}L_{uy} & -L_{\bar{u}\check{y}}^\T \bar{\textbf{X}}^{12} - \check{\textbf{X}}^{21}L_{\bar{u}\bar{y}} & -L_{\check{u}\check{y}}^\T \check{\textbf{X}}^{12} - \check{\textbf{X}}^{21}L_{\check{u}\check{y}} - \check{\textbf{X}}_{\check{p}}^{22} & -\check{\textbf{X}}^{21} L_{\check{u}w} + M_{z\check{y}}^\T \textbf{Y}^{21} \\
L_{uw}^\T & L_{\bar{u}w}^\T & L_{\check{u}w}^\T & -M_{zw}^\T\mathbf{Y}^{22} & -L_{uw}^\T \textbf{X}^{12} + \textbf{Y}^{12}M_{zy} & -L_{\bar{u}w}^\T\bar{\textbf{X}}^{12} + \textbf{Y}^{12}M_{z\bar{y}} & -L_{\check{u}w}^\T \check{\textbf{X}}^{12} + \textbf{Y}^{12}M_{z\check{y}} & M_{zw}^\T \textbf{Y}^{21} + \textbf{Y}^{12}M_{zw} + \textbf{Y}^{11}
\end{smallmatrix}\right] > \0
\end{equation}
\end{figure*}

\section{System Modeling}\label{problemformulation}
This section develops the dynamic model of the AC MG comprising multiple DGs, loads, and distribution lines. We present the local and distributed control laws, then formulate the closed-loop AC MG as a networked system suitable for dissipativity analysis.

\subsection{AC MG Topology}
The physical interconnection topology of an AC MG is represented as a directed connected graph $\mathcal{G}^p =(\mathcal{V},\mathcal{E})$. The vertex set $\mathcal{V}$ is partitioned into three disjoint subsets: DGs $\mathcal{D} \triangleq \{\Sigma_i^{DG}: i\in\N_N\}$, loads $\mathcal{L}=\{\Sigma_m^{Load}: m\in\N_M\}$, and lines $\mathcal{P} \triangleq \{\Sigma_l^{Line}: l\in\N_L\}$. DGs and loads are collectively denoted as $\mathcal{M} \triangleq \mathcal{D} \cup \mathcal{L}$ and interface with the AC MG through points of common coupling (PCCs), as illustrated in Fig. \ref{Networked}. Edge orientations represent reference directions for positive currents. Each line node $\mathcal{P}$ connects exactly two distinct nodes $\mathcal{M}$ through two directed edges.
The incidence matrix $\mathcal{B} \triangleq [\mathcal{B}_{kl}]_{k\in\N_{N+M},l\in\N_L}$ is defined such that $\mathcal{B}_{kl}$ equals $+1, -1$, or $0$ depending on whether $l$ is in $\mathcal{E}_k^+$, $\mathcal{E}_k^-$, or otherwise, respectively. Here 
$\mathcal{E}_k^+$ and $\mathcal{E}_k^-$ respectively denote out- and in-neighbors. We also define $\mathcal{E}_k \triangleq \mathcal{E}_k^+\cup \mathcal{E}_k^-$ (all neighbors), $\mathcal{N}^P_i \triangleq \{l:\mathcal{B}_{il}\neq0\}$ (lines connected to $\Sigma_i^{DG}$), $\mathcal{N}_m^P \triangleq \{l:\mathcal{B}_{ml} \neq 0\}$ (lines connected to $\Sigma_m^{load}$), and $\mathcal{N}^P_l\triangleq \{k:\mathcal{B}_{kl}\neq0\}$ (DGs/loads connected to $\Sigma_l^{line}$). 


\begin{figure}
    \centering
    \includegraphics[width=0.99\columnwidth]{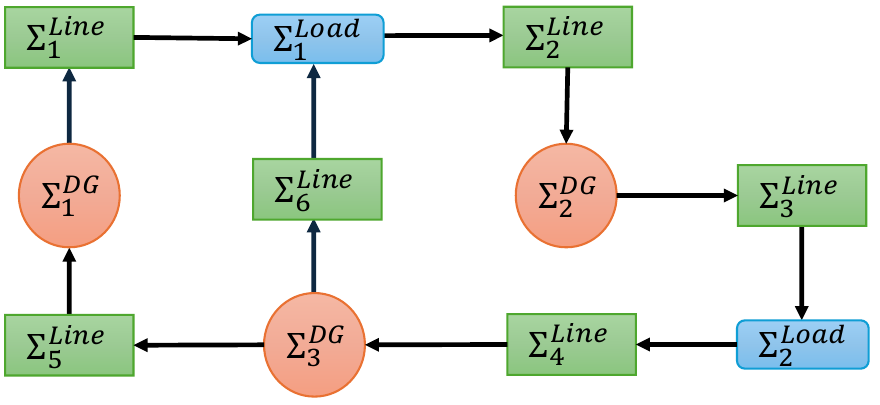}
    \caption{A representative diagram of the AC MG network. The sets $\mathcal{D}$, $\mathcal{L}$, and $\mathcal{P}$ are represented as DGs, distribution lines and loads.}
    \label{Networked}
\end{figure}

\subsection{Dynamic Model of a DG}
Each DG $\Sigma_i^{DG}$, as shown in Fig. \ref{AC_MG_scheme}(a), consists of a DC source, a voltage source converter (VSC), a series RLC filter, and a local load, which connects to multiple lines at PCC$_i$.

All three-phase electrical signals are represented in a common $dq$ reference frame rotating at the nominal frequency $\omega_0$ \cite{nahata2019passivity}. Any $dq$-frame variable, e.g., $V^{dq} \triangleq V^d + \mathrm{i}V^q$ can also be represented as a vector $V \triangleq \bm{V^d & V^q}^\T$ in $\R^2$. 

\begin{figure}
    \centering
    \includegraphics[width=0.99\columnwidth]{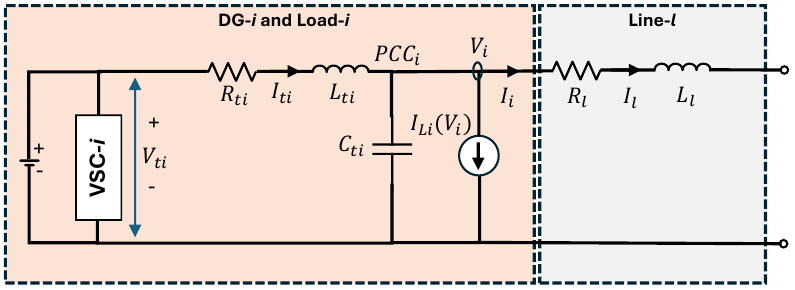}
    \caption{Electrical schematic of a DG $\Sigma_i^{DG}$ with a local load and a line $\Sigma_l^{Line}$.}
    \label{AC_MG_scheme}
\end{figure}

Applying Kirchhoff's law (KVL/KCL) at $\text{PCC}_i$ yields:
\begin{equation}
\begin{aligned}\label{eq:dgu_dynamics}
C_{ti}\frac{dV_i^{dq}}{dt} &= -C_{ti}\mathrm{i}\omega_0 V_i^{dq} + I_{ti}^{dq} - I_{Li}^{dq}(V_i^{dq}) - I_i^{dq}, \\
L_{ti}\frac{dI_{ti}^{dq}}{dt} &= -V_i^{dq} - \left(R_{ti} + L_{ti}\mathrm{i}\omega_0\right)I_{ti}^{dq} + V_{ti}^{dq},
\end{aligned}
\end{equation}
where $V_i^{dq}$ is the $\text{PCC}_i$ voltage, $I_{ti}^{dq}$ the filter current, $V_{ti}^{dq}$ the VSC voltage command, and $I_i^{dq}$ the net current injected into the AC MG given by:
\begin{align}\label{eq:current_injection}
I_i^{dq} = \sum_{l \in \mathcal{E}_i} \mathcal{B}_{il} I_l^{dq}. 
\end{align}

In \eqref{eq:dgu_dynamics}, the parameters $R_{ti}$, $L_{ti}$, and $C_{ti}$ denote the internal resistance, inductance, and filter capacitance.





The active and reactive power injections are
\cite{khan2023active}:
\begin{equation}\label{Eq:Active_Reactive_power}
\begin{aligned}
    P_i =V_i^d I_i^d + V_i^q I_i^q, \quad
    Q_i = V_i^q I_i^d - V_i^d I_i^q.
\end{aligned}
\end{equation}

\subsection{Load Model} 
The loads are modeled as ZIP (constant impedance (Z), constant current (I), and constant power (P)) components. At each $\Sigma_i^{DG}$, the local load current is: 
\begin{equation}
\label{eq:zip_load_local}
I_{Li}^{dq} = I_{Li}^{Z} + I_{Li}^{I} + I_{Li}^{P} \equiv Y_{Li}V_i^{dq} + \bar{I}_{Li}^{dq} + P_{Li}^{dq},
\end{equation}
where $Y_{Li}$, $\bar{I}_{Li}^{dq}$, and $P_{Li}^{dq}$ represent the constant admittance, current, and power components, respectively. 

For global load $\Sigma_m^{Load}$, as shown in Fig. \ref{AC_MG_scheme}(b), applying KCL at $\text{PCC}_m$ gives:
\begin{equation}\label{Eq:Load_Dynamic}
\frac{dV_m^{dq}}{dt} = -\mathrm{j}\omega_0 V_m^{dq} - \frac{1}{C_{tm}}\left(Y_{Lm}V_m^{dq} + \bar{I}_{Lm}^{dq}\right) - \frac{I_m^{dq}}{C_{tm}},
\end{equation}
where $I_m^{dq} = \sum_{l\in\E_m} \mathcal{B}_{ml}I_l^{dq}$ is the net current from connected lines. The state-space representation is:
\begin{align}\label{eq:load_statespace}
\dot{\check{x}}_m(t) &= \check{A}_m \check{x}_m(t) +  \check{B}_m \check{u}_m(t) + \check{B}_m \check{d}_m(t), 
\end{align}
with state $\check{x}_m \triangleq V_m$, exogenous input $\check{d}_m \triangleq \bar{I}_{Lm}$, coupling input $\check{u}_m \triangleq I_m$, and matrices:
\begin{equation}\label{eq:load_matrices}
\begin{aligned}
\check{A}_m \!\triangleq\! \begin{bmatrix} -\frac{Y_{Lm}}{C_{tm}} & \omega_0 \\ -\omega_0 & -\frac{Y_{Lm}}{C_{tm}} \end{bmatrix},\quad
\check{B}_m \!\triangleq\! \begin{bmatrix} -\frac{1}{C_{tm}} & 0 \\ 0 & -\frac{1}{C_{tm}} \end{bmatrix}.
\end{aligned}
\end{equation}

\begin{figure}
    \centering
    \includegraphics[width=0.7\columnwidth]{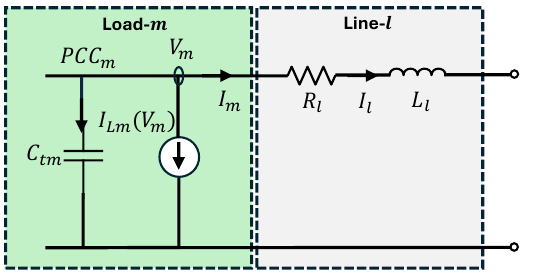}
    \caption{Electrical scheme of ZIP load in AC MG.}
    \label{fig:load_model}
\end{figure}

\subsection{Distribution Line Model}
Each line $\Sigma_l^{line}$, as shown in Fig. \ref{AC_MG_scheme}, is modeled as an RL circuit with resistance $R_l$ and inductance $L_l$ \cite{nahata2019passivity}. Applying KVL yields: 
\begin{align}\label{eq:line_dynamics}
\frac{dI_l^{dq}}{dt} = -\left(\frac{R_l}{L_l} + \mathrm{i}\omega_0\right)I_l^{dq} + \frac{1}{L_l}\sum_{i \in \mathcal{N}_l^P} \mathcal{B}_{il}V_i^{dq}.
\end{align}

The state space form is:
\begin{align}\label{eq:line_statespace}
\dot{\bar{x}}_l(t) = \bar{A}_l \bar{x}_l(t) + \bar{B}_l \bar{u}_l(t), 
\end{align}
with state $\bar{x}_l = I_l$, coupling input $\bar{u}_l = \sum_{k \in \mathcal{N}_l^P} \mathcal{B}_{kl}V_k^{dq}$, and
\begin{align}\label{eq:line_matrices}
\bar{A}_l &\triangleq \begin{bmatrix} -\frac{R_l}{L_l} & \omega_0 \\ -\omega_0 & -\frac{R_l}{L_l} \end{bmatrix}, \quad  \quad \bar{B}_l \triangleq \begin{bmatrix} \frac{1}{L_l} & 0 \\ 0 & \frac{1}{L_l} \end{bmatrix}. 
\end{align}



\section{Proposed Controllers}\label{Sec:Proposed_Controller}
This section introduces the hierarchical control architecture for the AC MG, consisting of local feedback controllers for voltage tracking and distributed consensus-based controllers for frequency coordination and proportional power distribution. The complete closed-loop system is then formulated as a networked system, as 
illustrated in Fig. \ref{fig.ACMG_topology}.

\subsection{Local Feedback Controllers}\label{Sec:Local_controller}
Each DG $\Sigma_i^{DG}$ implements a local feedback controller to track a specified PCC voltage reference $V_{i,ref}^{dq}$. To enable integral action, the $\Sigma_i^{DG}$ is augmented with an integral voltage error state $v_i^{dq}$ governed by:
\begin{align}\label{eq:integral_dynamics}
\dot{v}_i^{dq}(t) = e_i^{dq}(t) = V_i^{dq} - V_{i,ref}^{dq}.
\end{align}

The local control law employs a PI controller:
\begin{equation}\label{eq:local_controller}
\begin{aligned}
u_{iL}(t) \triangleq&\ K_i^P e_i^{dq}(t)  + K_i^I \int_0^t e_i^{dq}(\tau)  d\tau \\
=&\ K_{i0}x_i(t) - K_i^P V_{i,ref}^{dq}.
\end{aligned}
\end{equation}
where the augmented state vector of $\Sigma_i^{DG}$ is defined as
\begin{equation}\label{Eq:states}
    x_i \triangleq \bm{V_i^d & V_i^q & I_{ti}^d & I_{ti}^q & v_i^d & v_i^q & \tilde{\omega}_i}^\T,
\end{equation}
with $\tilde{\omega}_i$ representing the frequency error (introduced subsequently), $K_i^P$ and $K_i^I$ denoting the proportional and integral gain matrices, 
and $K_{i0} = \bm{K_i^P & \0_{2\times2} & K_i^I & \0_{2\times 1}}$.

This control architecture provides feedback from all augmented states to both $d$-axis and $q$-axis voltage commands, enabling cross-coupling when required for stability and performance.

\subsection{ Distributed Controllers for Frequency Synchronization and Power Sharing}
In islanded operation, each $\Sigma^{DG}_i$ must generate a local frequency reference $\omega_i(t)$ for $dq$-frame transformation and system-wide synchronization. To maintain frequency coherence while achieving droop-free power sharing, we incorporate frequency dynamics via normalized power consensus \cite{zuo2023asymptotic}. Defining the frequency error $\tilde\omega_i \triangleq \omega_i - \omega_0$, its dynamics are governed by:
\begin{equation}
    \label{eq:frequency_dynamics}
   \tau_i \frac{d\tilde\omega_i(t)}{dt} = -\tilde\omega_i + u_i^\Omega(t),
\end{equation}
where $\tau_i > 0$ represents a time constant capturing sensing, communication, and control delays, and $u_i^\Omega(t)$ is a distributed consensus control input:
\begin{equation}
\label{eq:frequency_control}
u_i^\Omega(t) \triangleq \sum_{j \in \mathcal{N}^C_i} K^\Omega_{ij} \left(P_{n,i}(t) - P_{n,j}(t)\right),
\end{equation} 
with $P_{n,i}(t) \triangleq P_i(t)/P_i^{\max}$ denoting the normalized active power and $[K_{ij}^\Omega]_{i,j\in\N_N}$ representing the distributed controller gains. 

To achieve proportional active and reactive power sharing, we employ additional distributed control inputs that directly regulate the VSC voltage $u_i^V(t) \triangleq V_{ti}$:
\begin{equation}\label{eq:distributed_controller}
\begin{aligned}
u_i^P(t) &= \sum_{j \in \mathcal{N}_i^C} K_{ij}^P \left(P_{i,n}(t) - P_{j,n}(t)\right),  \\
u_i^Q(t) &= \sum_{j \in \mathcal{N}_i^C} K_{ij}^Q \left(Q_{i,n}(t) - Q_{j,n}(t)\right), 
\end{aligned}
\end{equation}
where $[K_{ij}^P]_{i,j\in\N_N}$ and $[K_{ij}^Q]_{i,j\in\N_N}$ are the distributed gains, $Q_{i,n}(t) \triangleq Q_i(t)/Q_i^{\max}$ is the normalized reactive power, and $P_i^{\max}, Q_i^{\max}$ denote the maximum power ratings. The power measurements $P_i(t)$ and $Q_i(t)$ are obtained via \eqref{Eq:Active_Reactive_power}.

The total VSC control input $u_i^V(t) \triangleq V_{ti}$ combines three components:
$
u_i^V(t) \triangleq V_{ti}(t) = u_{iS} + u_{iL}(t) + u_{iG}(t),
$
where $u_{iS}$ is the steady-state controller (designed subsequently), $u_{iL}(t)$ is the local controller from \eqref{eq:local_controller}, and 
$u_{iG}(t) \triangleq \bm{u_i^P(t) & u_i^Q(t)}^\T$ represents the distributed controller from \eqref{eq:distributed_controller}.  Figure \ref{fig.ACMG_topology} illustrates this hierarchical control architecture. In the $dq$ frame:
\begin{equation}
    \label{eq:combined_control_dq}
    u_i^{Vdq}(t) \triangleq V_{ti}^{dq}(t) = u_{iS}^{dq} + u_{iL}^{dq}(t) + u_{iG}^{dq}(t).
\end{equation}

This droop-free control strategy eliminates the conventional trade-off between voltage regulation accuracy and power sharing precision by achieving coordination through direct communication rather than local droop characteristics.

\begin{figure*}[t]
    \centering
    \includegraphics[width=\textwidth]{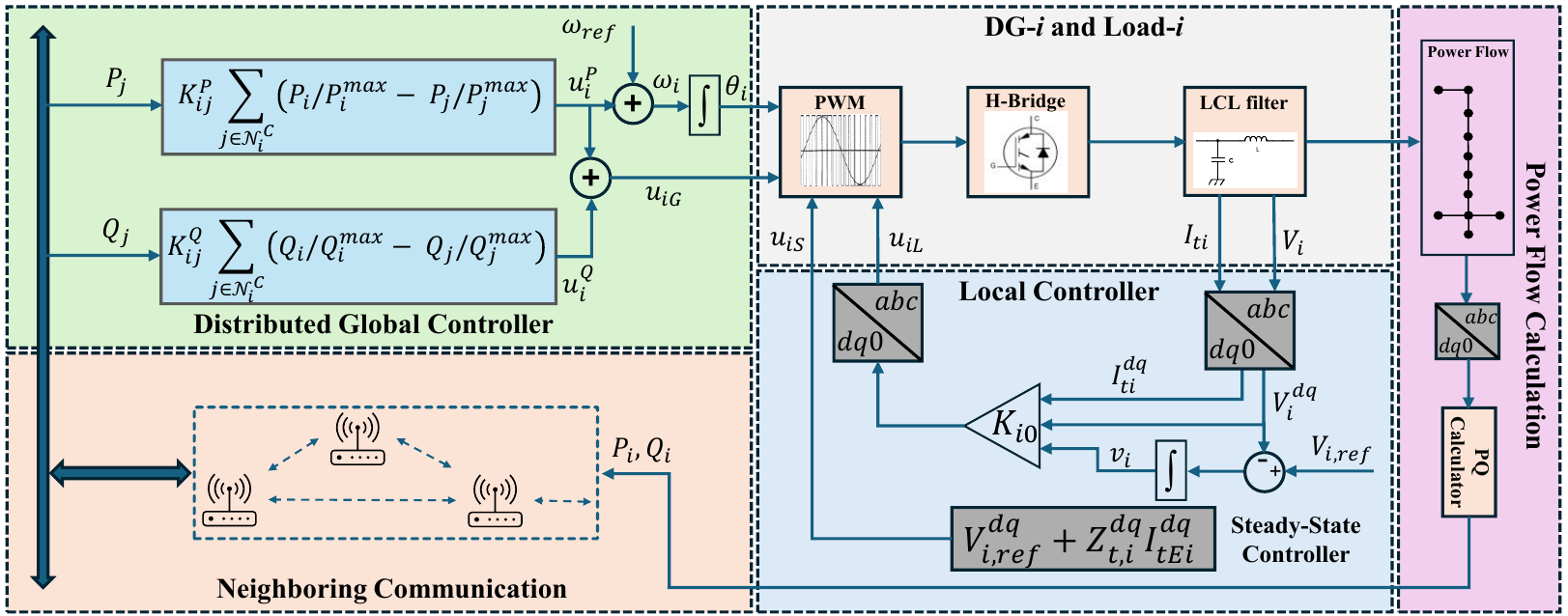}
    \caption{Hierarchical control architecture combining local and distributed controllers through physical and communication interconnections.}
    \label{fig.ACMG_topology}
\end{figure*}

\subsection{Closed-Loop Dynamics of the AC MG}
By combining \eqref{eq:dgu_dynamics}, \eqref{eq:integral_dynamics} and \eqref{eq:frequency_dynamics}, the complete dynamics of each $\Sigma_i^{DG},i\in\N_N$ are expressed as:
\begin{equation}\label{Eq:closedloop_dynamics}
\begin{aligned}
\frac{dV_i^{dq}}{dt} &= -\mathrm{i}(\omega_0 + \tilde\omega_i) V_i^{dq} + \frac{I_{ti}^{dq}}{C_{ti}} - \frac{I_{Li}^{dq}}{C_{ti}} - \frac{I_i^{dq}}{C_{ti}} + w_{vi}^{dq}(t),  \\
\frac{dI_{ti}^{dq}}{dt} &= -\frac{V_i^{dq}}{L_{ti}} - \big(\frac{R_{ti}}{L_{ti}} + \mathrm{i}(\omega_0 + \tilde\omega_i)\big)I_{ti}^{dq} + \frac{u_i^{Vdq}}{L_{ti}} + w_{ci}^{dq}(t),\\
\frac{dv_i^{dq}}{dt} &= V_i^{dq} - V_{i,ref}^{dq},\\
\frac{d\tilde\omega_i}{dt} &= -\frac{1}{\tau_i} \tilde\omega_i + \frac{1}{\tau_i}u_i^\Omega, 
\end{aligned}
\end{equation}
where the terms $I_i^{dq}$, $I_{Li}^{dq}$, $u_i^{Vdq}$, and $u_i^\Omega$ can all be substituted from \eqref{eq:current_injection}, \eqref{eq:zip_load_local},  \eqref{eq:combined_control_dq} and \eqref{eq:frequency_control}, respectively, and $w_{vi}^{dq}(t)$ and $w_{ci}^{dq}$ represent zero-mean disturbances accounting for model uncertainities and external impacts. Under the commonly adopted assumption $P_{Li}^{dq}=0$ (for simplicity), the dynamics \eqref{Eq:closedloop_dynamics} can be expressed in compact state-space form as: 
\begin{align}\label{eq:dgu_statespace}
\dot{x}_i(t) \!=\! A_i x_i(t) \negthickspace+\negthickspace B_i u_i(t) \negthickspace+\negthickspace E_id_i(t) \negthickspace+\negthickspace  F_i \xi_i(t) \negthickspace+\negthickspace g_i(x_i(t)), 
\end{align}
where $x_i(t)$ is the augmented state vector defined in \eqref{Eq:states}, $u_i(t) \triangleq \bm{u_i^V(t) & u_i^\Omega(t)}^\T$ represents the total control input, $d_i(t) \triangleq \bar{w}_i + w_i(t)$ is the exogenous input consisting of $\bar{w}_i \triangleq \bm{-\bar{I}_{Li}^{d} & -\bar{I}_{Li}^{q} & 0 & 0 & -V_{i,ref}^{d} & -V_{i,ref}^{q} & 0}^\T$ (known fixed disturbance) and  
$w_i(t) = \bm{w_{vi}^{d}(t) & w_{vi}^{q}(t) & w_{ci}^{d}(t) & w_{ci}^{q}(t) & 0 & 0 & 0}^\T$ (unknown zero-mean disturbance), $\xi_i(t) \triangleq I_i$ is the line coupling input from \eqref{eq:current_injection}, and $g_i(x_i(t))$ captures the nonlinear frequency coupling terms:
\begin{equation}\label{Eq:Nonlinearity_g}
g_i(x_i(t)) \triangleq \bm{-\tilde\omega_iV_i^{dq} & -\tilde\omega_iI_{ti}^{dq} & \0 & 0}^\T.
\end{equation}

The system matrices in \eqref{eq:dgu_statespace} are:
\begin{equation*}
A_i \triangleq \begin{bmatrix}
-\frac{Y_{Li}}{C_{ti}} & \omega_0 & \frac{1}{C_{ti}} & 0 & 0 & 0 & 0  \\
-\omega_0 & -\frac{Y_{Li}}{C_{ti}} & 0 & \frac{1}{C_{ti}} & 0 & 0 & 0  \\
-\frac{1}{L_{ti}} & 0 & -\frac{R_{ti}}{L_{ti}} & \omega_0 & 0 & 0 & 0  \\
0 & -\frac{1}{L_{ti}} & -\omega_0 & -\frac{R_{ti}}{L_{ti}} & 0 & 0 & 0 \\
1 & 0 & 0 & 0 & 0 & 0 & 0 \\
0 & 1 & 0 & 0 & 0 & 0 & 0\\
0 & 0 & 0 & 0 & 0 & 0 & -\frac{1}{\tau_i}
\end{bmatrix},
\end{equation*}
\begin{equation}\label{Eq:DG_matrices}
\begin{aligned}
B_i &\triangleq 
\begin{bmatrix}  0 & 0 & \frac{1}{L_{ti}} & 0 & 0 & 0 &  0 \\  0 & 0 & 0 &  \frac{1}{L_{ti}} & 0 & 0 & 0 \\ 0 & 0 & 0 & 0 & 0 & 0 & \frac{1}{\tau_i} \end{bmatrix}^\T, \\
E_i &\triangleq \diag(\bm{C_{ti}^{-1} & C_{ti}^{-1} & L_{ti}^{-1} & L_{ti}^{-1} & 1 & 1 &  \tau_i^{-1}}), \\
F_i &\triangleq 
\begin{bmatrix}  
    -\frac{1}{C_{ti}} & 0 & 0 & 0 & 0 & 0 & 0\\  0 & -\frac{1}{C_{ti}} & 0 & 0 & 0 & 0 & 0
\end{bmatrix}^\T.
\end{aligned}
\end{equation}

\subsection{Networked System Representation}\label{Sec:Networked_system}
This subsection formulates the closed-loop AC MG as a networked system by establishing interconnection relationships among DGs, lines, and loads through a static interconnection matrix.

\begin{figure}
    \centering
    \includegraphics[width=0.99\columnwidth]{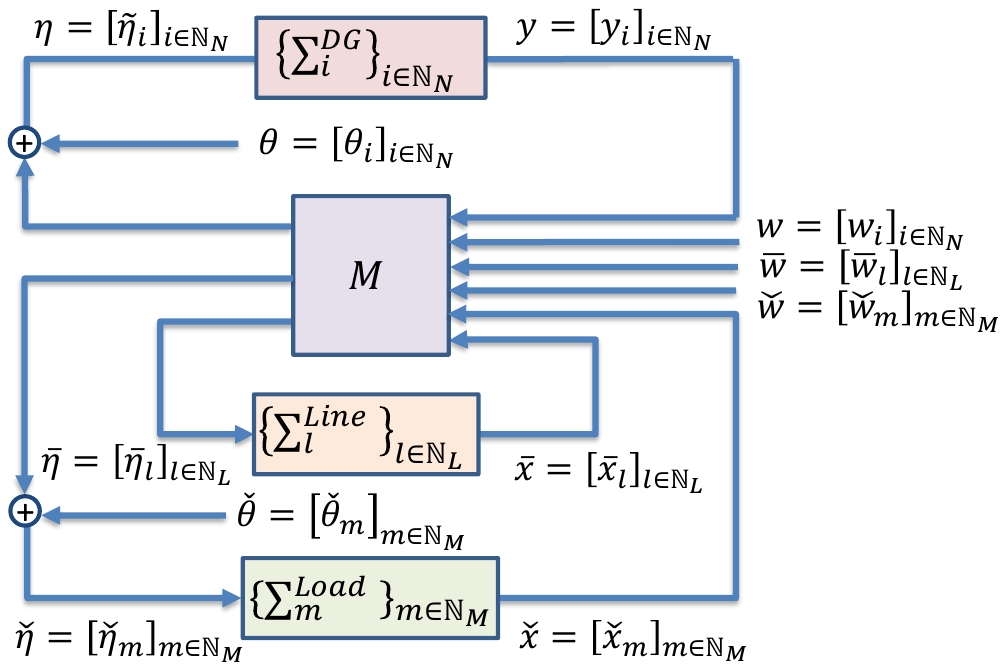}
    \caption{Networked system representation showing interconnections between DGs, lines, and loads AC MG.}
    \label{Networked_system_ACMG}
\end{figure}


\subsubsection{\textbf{DG Subsystems}}
The control input $u_i(t)$ in \eqref{eq:dgu_statespace} can be expressed as 
$$
u_i(t) = \bm{u_i^V(t) \\ u_i^\Omega(t)} 
= \bm{\I_{2 \times 2} \\ \0_{1\times 2}} 
(u_{iS} + u_{iL}(t)) + u_i^D(t), 
$$
where $u_{iS}$ is the steady-state controller, $u_{iL}(t)$ is the local feedback controller from \eqref{eq:local_controller}, and $u_i^D(t)$ combines the distributed controller from \eqref{eq:frequency_control} and \eqref{eq:distributed_controller}: 
\begin{align}\label{eq:distributed_controller_combined}
u_i^D(t) \triangleq&\ \bm{u_i^P(t) \\ u_i^Q(t) \\ u_i^\Omega(t)} = \sum_{j \in \bar{\mathcal{N}}_i^C} \hat{K}_{ij}  
\bm{P_j(t) \\ Q_j(t)},
\end{align}
where $\bar{\mathcal{N}}_i^C \triangleq \mathcal{N}_i^C \cup \{i\}$, and the gains matrices are defined as:
\begin{equation}\label{eq:controlconstraint}
\hat{K}_{ij} \triangleq 
\begin{cases}
 - \bar{K}_{ij}\diag({P_j^{\max}, Q_j^{\max}})^{-1}, \quad &j\in \mathcal{N}_i^C\\ 
\sum_{j\in \mathcal{N}_i^C} \bar{K}_{ij}\diag({P_i^{\max}, Q_i^{\max}})^{-1}, \quad &j = i.
\end{cases}
\end{equation}
with $\bar{K}_{ij}^\T \triangleq \bm{\diag^\T \bm{K_{ij}^P & K_{ij}^Q} & \bm{K_{ij}^\Omega & \0}^\T}^\T$.

To uniquely recover the distributed controller gains $K_{ij}^P, K_{ij}^Q, K_{ij}^\Omega$ from $\hat{K}_{ij}$, we impose two structural constraints:  
(i) each off-diagonal block $\hat{K}_{ij}$ ($j\neq i$) maintains the same structure as $\bar{K}_{ij}$; and 
(ii) each diagonal block $\hat{K}_{ij}$ ($j=i$) satisfies $\hat{K}_{ii} = -\sum_{j\in \mathcal{N}_i^C} \hat{K}_{ij} \diag(P_j^{\max}/P_i^{\max}, Q_j^{\max}/Q_i^{\max})$.

Substituting this $u_i(t)$ into \eqref{eq:dgu_statespace}, the closed-loop $\Sigma_i^{DG}$ dynamics become:
\begin{equation}
\label{Eq:closedloop_DG}
\begin{aligned}
\dot{x}_i(t) = 
&\hat{A}_ix_i(t) + \hat{B}_i\eta_i(t) + g_i(x_i(t)) + \theta_i,
\end{aligned}
\end{equation}
where $\hat{A}_i \triangleq A_i+ \bar{B}_iK_{i0}$, $\bar{B}_i \triangleq B_i\bm{\I_{2 \times 2} & \0_{2 \times 1}}^\T$, and 
$\hat{B}_i \triangleq \bm{B_i & F_i & E_i}$, $\theta_i \triangleq \bar{B}_i(u_{iS} - K_i^P \bm{ V_{i,ref}^d & V_{i,ref}^q}^\T) + E_i \bar{w}_i$, and  the effective control input is:
\begin{equation}\label{Eq:DG_Control_Equil}
\eta_i(t) \triangleq \bm{u_i^D(t) \\ \xi_i^\T(t) \\ w_i^\T(t)} = \bm{\sum_{j\in \mathcal{N}_i^C} K_{ij}y_j(t) \\ \sum_{l\in\E_i}\bar{C}_{il}\bar{x}_l(t) \\ w_i(t)},
\end{equation}
with $\xi_i \triangleq I_i =\sum_{l\in\E_i}\mathcal{B}_{il} I_l =  \sum_{l\in\E_i}\bar{C}_{il} \bar{x}_l$ representing the line coupling input where $\bar{C}_{il} \triangleq \mathcal{B}_{il} \I_{2}$.

The DG output vector $y_i$ capturing voltage and power injection is:
\begin{equation*}\label{Eq:Output}
    y_i = h_i(x_i, \eta_i) \triangleq 
    \bm{V_i^d & V_i^q & P_i & Q_i}^\T,
\end{equation*}
which can be computed using \eqref{Eq:states} and \eqref{Eq:Active_Reactive_power}. These outputs affect 
line inputs via $\bar{u}_l = \sum_{k \in \mathcal{N}_l^P} C_{lk} y_k$ where $C_{lk} \triangleq \bm{ \mathcal{B}_{kl}\I_2 & \0_{2\times 2}}$, and neighboring DG dynamics through \eqref{eq:distributed_controller_combined} where $
K_{ij} \triangleq \hat{K}_{ij} \bm{\0_{2\times 2} & \I_2} 
= \bm{\0_{3\times 2} & \hat{K}_{ij}}$.

\subsubsection{\textbf{Line Subsystems}} 
Each $\Sigma_l^{Line}$ follows the dynamics from \eqref{eq:line_statespace}:
\begin{equation}\label{Eq:closedloop_line}
    \dot{\bar{x}}_l(t) = \bar{A}_l\bar{x}_l(t) + \hat{\bar{B}}_l\bar{\eta}_l(t),
\end{equation}
where $\hat{\bar{B}}_l \triangleq \bm{\bar{B}_l & \bar{B}_l & \bar{E}_l}$ with $\bar{E}_l \triangleq (1/L_l)\I_2$. The effective input $\bar{\eta}_l$ is:
\begin{equation}\label{Eq:Line_Control_Equil}
    \bar{\eta}_l(t) = \bm{\sum_{k\in\mathcal{N}_l^P} C_{lk}y_k (t) \\ \sum_{m\in\mathcal{N}_l^P}\check{C}_{lm}\check{x}_m (t) \\
    \bar{w}_l(t)},
\end{equation}
where $C_{lk} \triangleq \bm{ \mathcal{B}_{kl}\I_2 & \0_{2 \times 2}}$, $\check{C}_{lm} \triangleq \mathcal{B}_{ml}\I_2$, and $\bar{w}_{l}(t)$ is a zero-mean disturbance included for robustness analysis. 

\subsubsection{\textbf{Load Subsystems}} Each load subsystem $\Sigma_m^{Load}$ follows dynamics from \eqref{eq:load_statespace}:
\begin{equation}\label{Eq:closedloop_Load}
    \dot{\check{x}}_m(t) = \check{A}_m\check{x}_m(t) + \hat{\check{B}}_m\check{\eta}_m(t) + \check{\theta}_m(t),
\end{equation}
where $\hat{\check{B}}_m \triangleq \bm{\check{B}_m & \check{E}_m}$ with $\check{E}_m \triangleq \check{B}_m$, and the effective input is:
\begin{equation}\label{Eq:Load_Control_Equil}
\check{\eta}_m(t) \triangleq \bm{\sum_{l\in\E_i}\bar{\check{C}}_{ml}\bar{x}_l(t) \\ \check{w}_m(t)
},
\end{equation}
where $\bar{\check{C}}_{ml} \triangleq \mathcal{B}_{ml}\I_2$, and $\check{w}_m(t)$ is a zero-mean disturbance. 

Vectorizing the subsystem dynamics and interconnections, the overall networked system representation is established with the interconnection matrix $M$ relating inputs and outputs as shown in Fig. \ref{Networked_system_ACMG}.

\section{Networked System Error Dynamics and Equilibrium Point Analysis}

This section analyzes the closed-loop AC MG equilibrium point and formulates the networked system error dynamics. 

\subsection{Formulating Networked System Error Dynamics}
We define error variables as deviations from constant equilibrium states. For DG, line, and load subsystems: $\tilde{V}_i^{dq} \triangleq V_i^{dq} - V_{Ei}^{dq}$, $\tilde{I}_{ti}^{dq} \triangleq I_{ti}^{dq} - I_{tEi}^{dq}$, $\tilde{v}_i^{dq} \triangleq v_i^{dq} - v_{Ei}^{dq}$, $\tilde{I}_i^{dq} \triangleq I_i^{dq} - I_{Ei}^{dq}$, $\tilde{I}_l^{dq} \triangleq I_l^{dq} - I_{El}^{dq}$, $\tilde{V}_m^{dq} \triangleq V_m^{dq} - V_{Em}^{dq}$, where we require $V_{Ei}^{dq} = V_{i,ref}$ and $v_{Ei}^{dq} = 0$ are required for voltage 
regulation, and $\tilde{\omega}_{Ei} = 0$ for frequency synchronization. Since error states:
$$
\tilde{x}_i \triangleq x_i - x_{Ei}, \ \ 
\tilde{\bar{x}}_l \triangleq \bar{x}_l - \bar{x}_{El}, 
\ \ \tilde{\check{x}}_m \triangleq \check{x}_m - \check{x}_{Em}, \ \ \forall i,l,m,
$$
their time derivatives are equal to those of the original states:
$$
\dot{\tilde{x}}_i = \dot{x}_i, \ \ 
\dot{\tilde{\bar{x}}}_l = \dot{\bar{x}}_l, \ \ 
\dot{\tilde{\check{x}}}_m = \dot{\check{x}}_m, \ \ \forall i,l,m.
$$

\subsubsection{\textbf{DG Error Subsystems}} Substituting $x_i = \tilde{x}_i + x_{Ei}$ and $\eta_i = \tilde{\eta}_i + \eta_{Ei}$ into \eqref{Eq:closedloop_DG}:
\begin{equation*}\label{Eq:DG_error}
    \dot{\tilde{x}}_i = \hat{A}_i(\tilde{x}_i + x_{Ei}) + \hat{B}_i(\tilde{\eta}_i + \eta_{Ei}) + g_i(\tilde{x}_i + x_{Ei}) + \theta_i.
\end{equation*}

From \eqref{Eq:Nonlinearity_g}, we have $g_i(\tilde{x}_i + x_{Ei}) = g_i(\tilde{x}_i) + \bar{g}_{i}(x_{Ei})\tilde{\omega}_i = g_i(\tilde{x}_i) + \bar{g}_{i}(x_{Ei}) e_7^\T \tilde{x}_i$ where $\bar{g}_{i}(x_{Ei}) \triangleq \bm{-V_{Ei}^q & -V_{Ei}^d & -I_{tEi}^q & -I_{tEi}^d & 0 & 0 & 0}^\T$ and $e_7$ is the unit vector in $\R^{7\times 1}$ with a $1$ in the last element. To avoid steady-state errors at the equilibrium point, we require:
\begin{equation}\label{Eq:Equil_condition}
    \hat{A}_i x_{Ei} + \hat{B}_i \eta_{Ei} + \theta_i = 0.
\end{equation}

This is enforced via steady-state control $u_{iS}$ in $\theta_i$. Under \eqref{Eq:Equil_condition}, the DG error dynamics become:
\begin{equation}\label{Eq:DG_error_dynamic}
    \dot{\tilde{x}}_i = \tilde{A}_i \tilde{x}_i + \hat{B}_i \tilde{\eta}_i + g_i(\tilde{x}_i),
\end{equation}
where $\tilde{A}_i \triangleq (\hat{A}_i+\bar{g}_{i}(x_{Ei}) e_7^\T)$, and the input error is:
\begin{equation}\label{Eq:DG_control_Error_Equil}
\begin{aligned}
\tilde{\eta}_i \triangleq\ \eta_i - \eta_{Ei} 
=\ \bm{\sum_{j\in \mathcal{N}_i^C} K_{ij}\tilde{y}_j(t) \\ \sum_{l\in\E_i}\bar{C}_{il}\tilde{\bar{x}}_l(t) \\ w_i(t)}.
\end{aligned}
\end{equation}

Since $\eta_{Ei} = \bm{u_{Ei}^D & \xi_{Ei}^\T & \0}^\T$ with $u_{Ei}^D = \0$. The frequency coupling nonlinearity error is:
\begin{equation}\label{Eq:DG_error_dynamics}
    g_i(\tilde{x}_i) \triangleq \bm{-\tilde{\omega}_i\tilde{V}_i^{dq} & -\tilde{\omega}_i\tilde{I}_{ti}^{dq} & \0 & 0}^\T.
\end{equation}

The DG output error $\tilde{y}_i \triangleq y_i - y_{Ei} 
=\bm{\tilde{V}_i^d & \tilde{V}_i^q & \tilde{P}_i & \tilde{Q}_i}^\T,$ where power errors are linearized around $(x_{Ei}, \eta_{Ei})$ using \eqref{Eq:Active_Reactive_power} is:
\begin{equation*}\label{Eq:Power_deviations}
\begin{aligned}
\tilde{P}_i &\triangleq P_i - P_{Ei} \approx V_{Ei}^d \tilde{I}_i^d + I_{Ei}^d \tilde{V}_i^d + V_{Ei}^q \tilde{I}_i^q + I_{Ei}^q \tilde{V}_i^q,\\
\tilde{Q}_i &\triangleq Q_i - Q_{Ei} \approx V_{Ei}^q \tilde{I}_i^d - V_{Ei}^d \tilde{I}_i^q + I_{Ei}^d \tilde{V}_i^q - I_{Ei}^q\tilde{V}_i^d.
\end{aligned}
\end{equation*}

Thus, the DG output error is expressed as:
\begin{equation}\label{Eq:Error_output}
    \tilde{y}_i = C_i^y  \tilde{x}_i + D_i^y \tilde{\eta}_i,
\end{equation}
where 
\begin{equation*}\label{Eq:output_matrices}
\scriptsize
\begin{aligned}
C_i^y \triangleq\ \begin{bmatrix}
1 & 0 & \0_{1 \times 5} \\
0 & 1 & \0_{1 \times 5} \\
I_{Ei}^d & I_{Ei}^q & \0_{1 \times 5} \\
-I_{Ei}^q & I_{Ei}^d & \0_{1 \times 5}
\end{bmatrix},
D_i^y \triangleq\ \begin{bmatrix}
\0_{1\times 3} & 0 & 0 & \0_{1\times 7}\\
\0_{1\times 3} & 0 & 0 & \0_{1\times 7} \\
\0_{1\times 3} & V_{Ei}^d & V_{Ei}^q & \0_{1\times 7} \\
\0_{1\times 3} & V_{Ei}^q & -V_{Ei}^d & \0_{1\times 7}
\end{bmatrix}.
\end{aligned}
\end{equation*}

\subsubsection{\textbf{Line Error Subsystems}}
Substituting $\bar{x}_l = \tilde{\bar{x}}_l + \bar{x}_{El}$ and $\bar{\eta}_l = \tilde{\bar{\eta}}_l + \bar{\eta}_{lE}$ into \eqref{Eq:closedloop_line} with $\bar{A}_l\bar{x}_{El} + \hat{\bar{B}}_l \bar{\eta}_{lE} = 0$ yields: 
\begin{equation}\label{Eq:line_error_dynamic}
    \dot{\tilde{\bar{x}}}_l = \bar{A}_l \tilde{\bar{x}}_l + \hat{\bar{B}}_l \tilde{\bar{\eta}}_l,
\end{equation} 
where line input error (using \eqref{Eq:Line_Control_Equil} and $\bar{w}_l$) is:
\begin{equation}\label{Eq:Line_control_Error_Equil}
    \tilde{\bar{\eta}}_l = \bm{\sum_{k\in\mathcal{N}_l^P \cap \mathcal{D}} C_{lk}\tilde{y}_k \\ \sum_{m\in\mathcal{N}_l^P \cap \mathcal{L}} \check{C}_{lm}\tilde{\check{x}}_m \\ \bar{w}_l}.
\end{equation}

\subsubsection{\textbf{Load Error Subsystems}}
Substituting $\check{x}_m = \tilde{\check{x}} + \check{x}_{Em}$ and $\check{\eta}_m = \tilde{\check{\eta}}_m + \check{\eta}_{mE}$ into \eqref{Eq:closedloop_Load}, requiring $\check{A}_m \check{x}_{Em} + \hat{\check{B}}_m \check{\eta}_{mE} + \check{\theta}_m= 0,$ yields:
\begin{equation}\label{Eq:load_error_dynamic}
    \dot{\tilde{\check{x}}}_m = \check{A}_m \tilde{\check{x}}_m + \hat{\check{B}}_m \tilde{\check{\eta}}_m,
\end{equation}
where load input error (using \eqref{Eq:Load_Control_Equil} and  $\check{w}_m$) is:
\begin{equation}\label{Eq:Load_control_Error_Equil}
    \tilde{\check{\eta}}_m = \bm{\sum_{l\in\E_m} \bar{\check{C}}_{ml}\tilde{\bar{x}}_l \\ \check{w}_m}.
\end{equation}

\subsubsection{\textbf{Performance Outputs and Disturbance Inputs}}
For dissipativity-based analysis, we define performance outputs: $z_i(t) \triangleq H_i \tilde{y}_i(t),$ for DGs with $H_i \triangleq \I_4$, $\bar{z}_l(t) \triangleq \bar{H}_l \tilde{\bar{x}}_l(t)$ for line with $\bar{H}_l = \I_2$, and $\check{z}_m (t) = \check{H}_m \tilde{\check{x}}_m(t)$ for loads with $\check{H}_m = \I_2$. Vectorizing yields:
\begin{equation*}\label{Eq:Error_Performance}
z = H \tilde{y}, \quad 
\bar{z} = \bar{H} \tilde{\bar{x}}, \quad 
\check{z} = \check{H} \tilde{\check{x}},
\end{equation*}
where $H \triangleq \diag(\bm{H_i}_{i\in\N_N})$, $\bar{H} \triangleq \diag(\bm{\bar{H}_l}_{l\in\N_L})$, and $\check{H} \triangleq \diag(\bm{\check{H}_m}_{m\in\N_M})$. Consolidating performance outputs:
$$
z_c 
\triangleq \bm{z^\T & \bar{z}^\T & \check{z}^\T}^\T 
= \bm{H_c & \bar{H}_c & \check{H}_c} 
\bm{\tilde{y}^\T & \tilde{\bar{x}}^\T & \tilde{\check{x}}^\T}^\T,
$$
where 
$$
H_c \triangleq \bm{H^\T \negthickspace&\negthickspace \0 \negthickspace&\negthickspace \0}^\T\negthickspace,
\bar{H}_c \triangleq \bm{\0 \negthickspace&\negthickspace \bar{H}^\T \negthickspace&\negthickspace \0}^\T\negthickspace,\check{H}_c \triangleq \bm{\0 \negthickspace&\negthickspace \0 \negthickspace&\negthickspace \check{H}^\T}^\T.
$$

Consolidating disturbance inputs:
$$
w_c 
\triangleq \bm{w & \bar{w} & \check{w}}^\T 
= 
\bm{E_c & \bar{E}_c & \check{E}_c}^\T 
\bm{\tilde{y} & \tilde{\bar{x}} & \tilde{\check{x}}}^\T,
$$
where 
$$
E_c \triangleq \bm{E & \0 & \0},\ 
\bar{E}_c \triangleq \bm{\0 & \bar{E} & \0},\ \check{E}_c \triangleq \bm{\0 & \0 & \check{E}},
$$
with:
$
E \triangleq \diag([\hat{E}_i]_{i\in\N_N})
$,
$ 
\bar{E} \triangleq \diag([\hat{\bar{E}}_l]_{l\in\N_L})
$, and 
$
\check{E} \triangleq \diag([\hat{\check{E}}_m]_{m\in\N_M})
$,
where
$$
\hat{E}_i \triangleq \bm{\0 & \0 & \I}^\T, \
\hat{\bar{E}}_l \triangleq \bm{\0 & \0 & \I}^\T, \ 
\hat{\check{E}}_m \triangleq \bm{\0 & \I }^\T.
$$

\subsubsection{\textbf{Networked System Error Dynamics Representation}}
For DG error subsystems, vectorizing \eqref{Eq:DG_error_dynamic}, \eqref{Eq:DG_control_Error_Equil}, and \eqref{Eq:Error_output} across all $i\in\N_N$ yields:
\begin{equation*}\label{Eq:Error_Vectorized_DG}
\begin{cases}
    \dot{\tilde{x}} = \tilde{A}\tilde{x} + \hat{B}\tilde{\eta} + G(\tilde{x}),\quad
    \tilde{\eta} = \tilde{K}\tilde{y} + \bar{C}\tilde{\bar{x}} + E w,\\
    \tilde{y} = C^y \tilde{x} + D^y \tilde{\eta}, 
\end{cases}
\end{equation*}
where we define $\tilde{x} \triangleq [\tilde{x}_i^\top]_{i \in \N_N}^\T$, 
$\tilde{\eta} \triangleq [\tilde{\eta}_i^\T]_{i \in \N_N}^\top$, $\tilde{y} \triangleq [\tilde{y}_i^\top]_{i \in \N_N}^\T$, $\tilde{\bar{x}} \triangleq [\tilde{\bar{x}}_l^\T]_{l \in \N_L}^\top$, the matrices: 
$\tilde{A} \triangleq \diag([\tilde{A}_i]_{i \in \N_N})$, $\bar{C} \triangleq \bm{\0 & \bar{C}_{il} & \0}^\T_{i\in\N_N, l\in\N_L}$,
$C^y \triangleq \diag([C_i^y]_{i \in \N_N})$, 
$D^y \triangleq \diag([D_i^y]_{i \in \N_N})$, and $G(\tilde{x}) \triangleq [g_i^\T(\tilde{x}_i)]_{i \in \N_N}^\top$.

For line error subsystems, vectorizing \eqref{Eq:line_error_dynamic} and \eqref{Eq:Line_control_Error_Equil} across all $l\in\N_L$ yields:
\begin{equation*}\label{Eq:Error_Vectorized_Line}
\dot{\tilde{\bar{x}}} = \bar{A}\tilde{\bar{x}} + \hat{\bar{B}}\tilde{\bar{\eta}}, \quad
\tilde{\bar{\eta}} = C\tilde{y} + \check{C}\tilde{\check{x}} + \bar{E}\bar{w},
\end{equation*}
where we define  
$\tilde{\bar{\eta}} \triangleq [\tilde{\bar{\eta}}_l^\T]_{l \in \N_L}^\T$ and  $\tilde{\check{x}} \triangleq [\tilde{\check{x}}_m^\T]_{m \in \N_M}^\T$, $C \triangleq \bm{C_{lk} & \0 & \0}^\T_{l\in\N_L, k\in\N_N},  
\check{C} \triangleq \bm{\0 & \check{C}_{lm} & \0}^\T_{l\in\N_L, m\in\N_M}$.

For load error subsystems, vectorizing \eqref{Eq:load_error_dynamic} and \eqref{Eq:Load_control_Error_Equil} across all $m\in\N_M$ yields:
\begin{equation*}\label{Eq:Error_Vectorized_Load}
\dot{\tilde{\check{x}}} = \check{A}\tilde{\check{x}} + \hat{\check{B}}\tilde{\check{\eta}}, \quad
\tilde{\check{\eta}} = \bar{\check{C}}\tilde{\bar{x}} + \check{E}\check{w},
\end{equation*}
where 
$\tilde{\check{\eta}} \triangleq [\tilde{\check{\eta}}_m^\T]_{m \in \N_M}^\T$, $\bar{\check{C}} = \bm{\bar{\check{C}}_{ml} & \0}^\T_{m\in\N_M, l\in\N_L}$.

The interconnection among error subsystems, shown in Fig. \ref{Error_Networked_system_ACMG}, is:
\begin{equation*}
    \bm{\tilde{\eta}^\T & \tilde{\bar{\eta}}^\T & \tilde{\check{\eta}}^\T &  z_c^\T}^\T = \tilde{M} \bm{\tilde{y}^\T & \tilde{\bar{x}}^\T & \tilde{\check{x}}^\T & w_c^\T}^\T,
\end{equation*}
where $\tilde{M}$ is the connection matrix:
\begin{equation}
    \tilde{M} \triangleq
\begin{bmatrix}\label{Eq:Error_MMmatrix}
\tilde{K} & \bar{C} & \0 & E_c \\
C & \0 & \check{C} & \bar{E}_c \\
\0 & \bar{\check{C}} & \0 & \check{E}_c \\
H_c & \bar{H}_c & \check{H}_c & \0
\end{bmatrix}.
\end{equation}

\begin{figure}
    \centering
    \includegraphics[width=0.98\columnwidth]{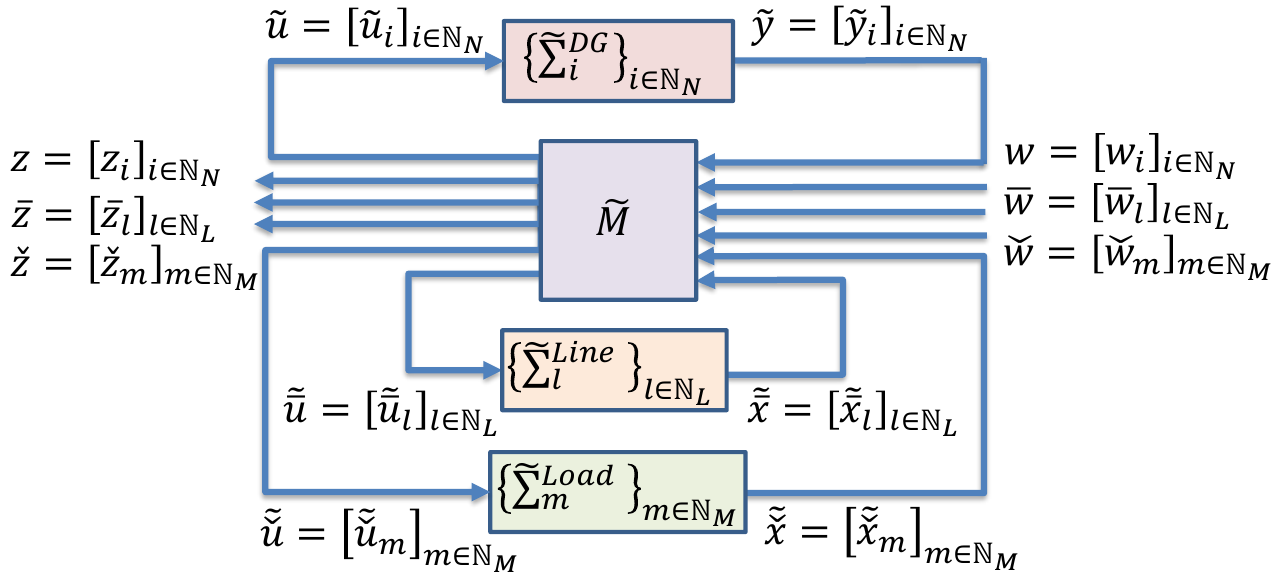}
    \caption{Error networked system configuration showing interconnections between DGs, lines, and loads error dynamics in AC MG.}
    \label{Error_Networked_system_ACMG}
\end{figure}

\subsection{Equilibrium Point Analysis of the AC MG}
The equilibrium point analysis enables steady-state controller design to ensure the AC MG achieves simultaneous voltage regulation, frequency synchronization, and proportional power sharing simultaneously.

The incidence matrix $\mathcal{B}\in\R^{(N+M)\times L}$ defined earlier is partitioned as $\mathcal{B} = \bm{\mathcal{B}_D^\T & \mathcal{B}_L^\T}^\T$, where $\mathcal{B}_D \in \R^{N\times L}$  corresponds to DG 
connections and $\mathcal{B}_L \in \R^{M \times L}$ to load connections. The system parameter matrices are $C_t \triangleq \diag([C_{ti}]_{i\in\N_N})$, $L_t \triangleq \diag([L_{ti}]_{i\in\N_N})$, $R_t \triangleq \diag([R_{ti}]_{i\in\N_N})$, $Y_L \triangleq \diag([Y_{Li}]_{i\in\N_N})$, $R \triangleq \diag([R_l]_{l\in\N_L})$, $L \triangleq \diag([L_l]_{l\in\N_L})$, $\check{C}_t \triangleq \diag([C_{tm}]_{m\in\N_M})$, and $\check{Y}_L \triangleq \diag([Y_{Lm}]_{m\in\N_M})$. The complex impedance and admittance matrices are $Z^{dq} \triangleq \diag([R_l + \mathrm{i}\omega_0L_l]_{l\in\N_L})$, $Z_t^{dq} \triangleq \diag([R_{ti}+\mathrm{i}\omega_0 L_{ti}]_{i\in\N_N})$, and $\check{Y}^{dq} \triangleq \diag([Y_{Lm} + \mathrm{i}\omega_0C_{tm}]_{m\in\N_M})$, with $\Omega_0 \triangleq \diag([\mathrm{i}\omega_0]_{i\in\N_N})$.  

\begin{lemma}\label{Lm:Equilibrium_point}
Assuming zero-mean unknown disturbances vanish, i.e., $w_i(t)=0, \forall i\in\N_N,\ \bar{w}_l(t)=0, \forall l\in\N_L,\ \check{w}_m=0, \forall m \in\N_M$, for a given reference voltage vector $V_r \triangleq [V_{i,ref}^{dq}]_{i\in\N_N}$, constant current load vectors $\bar{I}_L \triangleq [\bar{I}_{Li}^{dq}]_{i\in\N_N}$ and $\check{\bar{I}}_L \triangleq [\bar{I}_{Lm}^{dq}]_{m\in\N_M}$, under a fixed control input $u_E^{Vdq} \triangleq [u_{Ei}^{Vdq}]_{i\in\N_N}$ defined as:
\begin{equation}\label{Eq:Equil_control}
     u_E^{Vdq} = V_r + Z_t^{dq} \big(Y_LV_r+ \Omega_0C_t V_r + \bar{I}_L + \mathcal{B}_DI_E^{dq}\big),
\end{equation}
there exists a unique equilibrium point for the AC MG characterized by reference voltage vector $V_r$, constant current load vectors $\bar{I}_L$ and $\check{\bar{I}}_L$, given by: 
\begin{subequations}\label{Eq:Equil}
\begin{align}
    V_E^{dq} &= V_r, \label{Eq:Equil_voltage}\\
    \tilde{\omega}_E &= \0, \label{Eq:Equil_frequency}\\
     I_{tE}^{dq} &= (Z_t^{dq})^{-1}(u_E^{dq} - V_E^{dq}), \label{Eq:Equil_DGcurrent}\\
    I_E^{dq} &= (Z^{dq})^{-1} \mathcal{B}_D^\T V_E^{dq} + (Z^{dq})^{-1}\mathcal{B}_L^\T \check{V}_E^{dq}, \label{Eq:Equil_linecurrent}\\
    \check{V}_E^{dq} &= -(\check{Y}^{dq})^{-1}(\check{\bar{I}}_L^{dq} + \mathcal{B}_L I_E^{dq}), \label{Eq:Equil_load}
\end{align}
\end{subequations}
where we define the equilibrium state vectors $V_E^{dq}\triangleq [V_{Ei}^{dq}]_{i\in\N_N}\in\R^{2N}$ (DG voltages), $\tilde{\omega}_E \triangleq [\tilde{\omega}_{Ei}]_{i\in\N_N}\in \R^N$ (frequency deviations), $I_{tE}^{dq} \triangleq [I_{tEi}^{dq}]_{i\in\N_N}\in\R^{2N}$ (DG filter currents), $I_E^{dq} \triangleq [I_{El}^{dq}]_{l\in\N_N}\in\R^{2L}$ (line currents), $\check{V}_E^{dq} \triangleq [V_{mE}^{dq}]_{m\in\N_M} \in \R^{2M}$ (load voltages).
\end{lemma}

\begin{proof}
At equilibrium, all state derivatives in the closed-loop system vanish. We establish the equilibrium conditions by analyzing each subsystem type. 

\subsubsection{\textbf{DG Subsystems}} The equilibrium state of the closed-loop DG dynamics of \eqref{Eq:Equil_condition} satisfies:
\begin{equation*}
    \hat{A}_i x_{Ei} + \hat{B}_i \eta_{Ei} + \theta_i = \0,
\end{equation*}
where $x_{Ei} \triangleq \bm{V_{Ei}^d & V_{Ei}^q & I_{tEi}^d & I_{tEi}^q & v_{Ei}^d & v_{Ei}^q & \tilde{\omega}_{Ei}}^\T$ represents the equilibrium state, $\eta_{Ei}$ represents the equilibrium effective input from \eqref{Eq:DG_Control_Equil}, and using $\hat{A}_i$ and $\hat{B}_i$ defined in \eqref{Eq:closedloop_DG}, we obtain  from the respective rows:
\begin{subequations}
\begin{align}
    &-\frac{Y_{Li}}{C_{ti}}V_{Ei}^{dq} \negthickspace-\negthickspace \mathrm{i}(\omega_0\negthickspace+\negthickspace\tilde{\omega}_{Ei})V_{Ei}^{dq} \negthickspace+\negthickspace \frac{I_{tEi}^{dq}}{C_{ti}} \negthickspace-\negthickspace \frac{\bar{I}_{Li}^{dq}}{C_{ti}} \negthickspace-\negthickspace \frac{I_{Ei}^{dq}}{C_{ti}} = 0, \label{eq:equi_voltage}\\
    &-\frac{V_{Ei}^{dq}}{L_{ti}} -\Big(\frac{R_{ti}}{L_{ti}}+\mathrm{i}(\omega_0+\tilde{\omega}_{Ei})\Big) I_{tEi}^{dq}  + \frac{u_{Ei}^{Vdq}}{L_{ti}} = 0, \label{eq:equi_current}\\
    &- V_{i,ref}^{dq} + V_{Ei}^{dq}  = 0, \label{eq:equi_integral}\\
    &-\frac{1}{\tau_i}\tilde{\omega}_{Ei} + \frac{1}{\tau_i}u_{Ei}^{\Omega} = 0. \label{eq:equi_frequency}
\end{align}
\end{subequations}

From \eqref{eq:equi_integral}, we directly obtain:
\begin{equation}\label{Eq:Vectorized_Equil_voltage}
   V_{Ei}^{dq} = V_{i,ref}^{dq},
\end{equation}
which vectorizes to establish \eqref{Eq:Equil_voltage}. From \eqref{eq:equi_frequency}, we have $\tilde{\omega}_{Ei} = u_{Ei}^\Omega$, where under the normalized power sharing control law \eqref{eq:frequency_control}:
\begin{equation*}
    u_{Ei}^\Omega = \sum_{j\in\mathcal{N}_i^C}K_{ij}^\Omega(P_{n,Ei} - P_{n,Ej}),
\end{equation*}
with $P_{n,Ei} = P_{Ei}/P_i^{\max}$ representing the normalized active power of $\Sigma_i^{DG}$. At equilibrium $P_{n,Ei} = P_{n,Ej}$ for all $j\in\mathcal{N}_i^C$, which implies that each term in the summation vanishes. Therefore, $u_{Ei}^\Omega = 0$ and thus $\tilde{\omega}_{Ei} = 0$ for all $i\in\N_N$, establishing \eqref{Eq:Equil_frequency}.

With $\tilde{\omega}_{Ei} = 0$, \eqref{eq:equi_current} simplifies to:
\begin{equation*}
    -\frac{V_{Ei}^{dq}}{L_{ti}} -\Big(\frac{R_{ti}}{L_{ti}}+\mathrm{i}\omega_0\Big) I_{tEi}^{dq}  + \frac{u_{Ei}^{Vdq}}{L_{ti}} = 0,
\end{equation*}
which by multiplying by $L_{ti}$ and rearranging yields:
\begin{equation*}
    u_{Ei}^{Vdq} = V_{Ei}^{dq} + (R_{ti} + \mathrm{i}\omega_0L_{ti})I_{tEi}^{dq},
\end{equation*}
where defining the $i$-th diagonal element of $Z_t$ as $Z_{t,i}^{dq} \triangleq R_{ti} + \mathrm{i}\omega_0L_{ti}$, gives:
\begin{equation}\label{Eq:Equilibrium_Control_Input}
     u_{Ei}^{Vdq} = V_{Ei}^{dq} + Z_{t,i}^{dq} I_{tEi}^{dq}.
\end{equation}

Solving for the equilibrium filter current and substituting $V_{Ei}^{dq} = V_{i,ref}^{dq}$ from \eqref{Eq:Vectorized_Equil_voltage}:
\begin{equation*}
     I_{tEi}^{dq} = (Z_{t,i}^{dq})^{-1}(u_{Ei}^{Vdq} - V_{i,ref}^{dq}),
\end{equation*}
which vectorizes for $i\in\N_N$ yield:
\begin{equation}\label{Eq:Vectorized_Equil_DGcurrent}
    I_{tE}^{dq} = (Z_t^{dq})^{-1}(u_E^{Vdq} - V_r),
\end{equation}
establishing \eqref{Eq:Equil_DGcurrent}.

With $\tilde{\omega}_{Ei} = 0$, \eqref{eq:equi_voltage} can be written by multiplying by $C_{ti}$ and substituting $I_{Ei}^{dq} = \sum_{l\in\E_i}\mathcal{B}_{il}I_{El}^{dq}$:
\begin{equation*}
    -Y_{Li}V_{Ei}^{dq} - \mathrm{i}\omega_0C_{ti}V_{Ei}^{dq} + I_{tEi}^{dq} - \bar{I}_{Li}^{dq} - \sum_{l\in\E_i}\mathcal{B}_{il}I_{El}^{dq} = 0,
\end{equation*}
which rearranges to gives:
\begin{equation*}
    I_{tEi}^{dq} = (Y_{Li} + \mathrm{i}\omega_0C_{ti})V_{Ei}^{dq} + \bar{I}_{Li}^{dq} + \sum_{l\in\E_i}\mathcal{B}_{il}I_{El}^{dq}.
\end{equation*}

Substituting the expression for $I_{tEi}^{dq}$ from \eqref{Eq:Vectorized_Equil_DGcurrent} and $V_{Ei}^{dq} = V_{i,ref}^{dq}$ from \eqref{Eq:Vectorized_Equil_voltage} yields:
\begin{equation*}
    (Z_{t,i}^{dq})^{-1}(u_{Ei}^{Vdq} - V_{i,ref}^{dq}) \negthickspace=\negthickspace (Y_{Li} \!+\! \mathrm{i}\omega_0C_{ti})V_{i,ref}^{dq} \!+\! \bar{I}_{Li}^{dq} \!+\! \sum_{l\in\E_i}\mathcal{B}_{il}I_{El}^{dq},
\end{equation*}
which by multiplying both sides by $Z_{t,i}^{dq}$ and rearranging gives:
\begin{equation*}
    u_{Ei}^{Vdq} = V_{i,ref}^{dq} + Z_{t,i}^{dq}\Big((Y_{Li} + \mathrm{i}\omega_0C_{ti})V_{i,ref}^{dq} + \bar{I}_{Li}^{dq} + \sum_{l\in\E_i}\mathcal{B}_{il}I_{El}^{dq}\Big).
\end{equation*}

Note that $\sum_{l\in\E_i}\mathcal{B}_{il}I_{El}^{dq} = \bm{\mathcal{B}_D I_{El}^{dq}}_i$ is the $i$-th component of the vector $\mathcal{B}_D I_{El}^{dq}$. Vectorizing for all $i\in\N_N$ with $\Omega_0 \triangleq \diag(\bm{\mathrm{i}\omega_0}_{i\in\N_N})$:
\begin{equation*}
    u_E^{Vdq} = V_r + Z_t^{dq} \Big((Y_L + \Omega_0C_t)V_r + \bar{I}_L + \mathcal{B}_DI_E^{dq}\Big),
\end{equation*}
which can be rewritten as:
\begin{equation*}
    u_E^{Vdq} = V_r + Z_t^{dq} \big(Y_LV_r+ \Omega_0C_t V_r + \bar{I}_L + \mathcal{B}_DI_E^{dq}\big),
\end{equation*}
establishing \eqref{Eq:Equil_control}.

\subsubsection{\textbf{Line Subsystems}} The equilibrium state of the line dynamics from \eqref{Eq:line_error_dynamic} satisfies:
\begin{equation*}
    \bar{A}_l\bar{x}_{El} + \hat{\bar{B}}_l \bar{\eta}_{lE} = 0,
\end{equation*}
where $\bar{x}_{El} = I_{El}^{dq}$ and $\bar{\eta}_{lE}$ is the equilibrium effective input of \eqref{Eq:Line_control_Error_Equil} given by:
\begin{equation*}
    \bar{\eta}_{lE} = \bm{\sum_{k\in\mathcal{N}_l^P\cap\mathcal{D}} C_{lk} y_{Ek} \\ \sum_{m\in\mathcal{N}_l^P\cap \mathcal{L}}\check{C}_{lm} \check{x}_{Em}\\
    \0}.
\end{equation*}

We can extract voltages as
$
\sum_{k\in\mathcal{N}_l^P} C_{lk} y_{Ek} = \sum_{k\in\mathcal{N}_l^P} \mathcal{B}_{kl} V_{Ek}^{dq},
$
$
\sum_{m\in\mathcal{N}_l^P} \check{C}_{lm} \check{x}_{Em} = \sum_{m\in\mathcal{N}_l^P} \mathcal{B}_{ml} V_{mE}^{dq},
$
respectively. Using $\bar{A}_l$ and $\bar{B}_l$ from \eqref{eq:line_matrices} gives:
\begin{equation*}
    -\Big(\frac{R_l}{L_l}+\mathrm{i}\omega_0\Big)I_{El}^{dq} + \frac{1}{L_l}\sum_{k\in\mathcal{N}_l^P}\mathcal{B}_{kl}V_{Ek}^{dq} = 0,
\end{equation*}
which by multiplying by $L_l$ and rearranging:
\begin{equation*}
\big(R_l+\mathrm{i}\omega_0L_l\big)I_{El}^{dq} = \sum_{k\in\mathcal{N}_l^P}\mathcal{B}_{kl}V_{Ek}^{dq},
\end{equation*}
where defining the $l$-th diagonal element of $Z^{dq}$ as $Z_l^{dq} \triangleq R_l + \mathrm{i}\omega_0L_l$ gives:
\begin{equation*}\label{Eq:line_Equil}
    I_{El}^{dq} = (Z_l^{dq})^{-1}\sum_{k\in\mathcal{N}_l^P}\mathcal{B}_{kl}V_{Ek}^{dq}.
\end{equation*}

This summation can be decomposed as 
\begin{align*}
\sum_{k\in\mathcal{N}_l^P}\mathcal{B}_{kl}V_{Ek}^{dq} &= \sum_{i\in\mathcal{N}_l^P\cap\mathcal{D}}\mathcal{B}_{il}V_{Ei}^{dq} + \sum_{m\in\mathcal{N}_l^P\cap\mathcal{L}}\mathcal{B}_{ml}V_{mE}^{dq}\\&= \bm{\mathcal{B}_D^\T V_E^{dq}}_l + \bm{\mathcal{B}_L^\T V_{mE}^{dq}}_l,
\end{align*}
where the first term corresponds to DG contributions and the second to load contributions. Therefore:
\begin{equation*}
    I_{El}^{dq} = (Z_l^{dq})^{-1}\Big(\bm{\mathcal{B}_D^\T V_E^{dq}}_l + \bm{\mathcal{B}_L^\T V_{mE}^{dq}}_l\Big),
\end{equation*}
which vectorizes for all $l\in\N_L$ yields: 
\begin{equation*}
    I_{El}^{dq} = (Z^{dq})^{-1} \mathcal{B}_D^\T V_E^{dq} + (Z^{dq})^{-1} \mathcal{B}_L^\T V_{mE}^{dq},
\end{equation*}
establishing \eqref{Eq:Equil_linecurrent}.

\subsubsection{\textbf{Load Subsystems}} The equilibrium state of the load dynamics from \eqref{Eq:load_error_dynamic} satisfies:
\begin{equation*}
\check{A}_m \check{x}_{Em} + \hat{\check{B}}_m \check{\eta}_{mE} + \check{\theta}_m= 0,
\end{equation*}
where $\check{x}_{Em} = V_{mE}^{dq}$ and $\check{\eta}_{mE}$ is the equilibrium effective input of \eqref{Eq:Load_control_Error_Equil} given by:
\begin{equation*}
    \check{\eta}_{mE} = \bm{\sum_{l\in\E_m} \bar{\check{C}}_{ml}\bar{x}_{El} \\ \0}.
\end{equation*}

We obtain
$
\sum_{l\in\E_m} \bar{\check{C}}_{ml}\bar{x}_{El} = \sum_{l\in\E_m} \mathcal{B}_{ml}I_{El}^{dq},
$
which using $\check{A}_m$ and $\hat{B}_m$ from \eqref{eq:load_matrices} gives:
\begin{equation*}
    -\Big(\frac{Y_{Lm}}{C_{tm}} + \mathrm{i}\omega_0\Big) V_{mE}^{dq} - \frac{1}{C_{tm}}\Big(\sum_{l\in\E_m} \mathcal{B}_{ml}I_{El}^{dq} + \bar{I}_{Lm}^{dq}\Big) = 0.
\end{equation*}

Multiplying by $C_{tm}$ and rearranging yields:
\begin{equation*}
    -(Y_{Lm} + \mathrm{i}\omega_0 C_{tm}) V_{mE}^{dq}  = \sum_{l\in\E_m}\mathcal{B}_{ml}I_{El}^{dq} + \bar{I}_{Lm}^{dq},
\end{equation*}
where defining the $m$-th diagonal element of $Y^{dq}$ as 
$Y_m^{dq} \triangleq Y_{Lm} + \mathrm{i}\omega_0C_{tm}$ gives:
\begin{equation*}\label{Eq:voltage_load}
    V_{mE}^{dq} = -(Y_m^{dq})^{-1}\Big(\bar{I}_{Lm}^{dq} + \sum_{l\in\E_m}\mathcal{B}_{ml}I_{El}^{dq}\Big).
\end{equation*}

Note that $\sum_{l\in\E_m}\mathcal{B}_{ml}I_{El}^{dq} = \bm{\mathcal{B}_L I_{El}^{dq}}_m$ is the $m$-th component of $\mathcal{B}_L I_{El}^{dq}$. Vectorizing for all $m\in\N_M$ yields:
\begin{equation*}
    V_{mE}^{dq} = -(Y^{dq})^{-1}\Big(\bar{I}_{Lm}^{dq} + \mathcal{B}_L I_E^{dq}\Big),
\end{equation*}
establishing \eqref{Eq:Equil_load}.

This completes the proof. The coupled equations \eqref{Eq:Equil_control} and \eqref{Eq:Equil} uniquely determine all equilibrium state variables.
\end{proof}

\begin{remark}\label{Rm:powersharing}
At equilibrium, proportional power sharing requires:
\begin{equation}\label{Eq:powersharing}
\begin{aligned}
    \frac{P_{Ei}}{P_i^{\max}} &= P_s \Longleftrightarrow P_{Ei} = P_i^{\max} P_s, \ \forall i\in\N_N,\\
    \frac{Q_{Ei}}{Q_i^{\max}} &= Q_s \Longleftrightarrow Q_{Ei} = Q_i^{\max} Q_s, \ \forall i\in\N_N,
\end{aligned}
\end{equation}
which vectorizes as $P_E = P_{\max} \mathbf{1}_N P_s$ and $Q_E = Q_{\max} \mathbf{1}_N Q_s$, where $P_{\max}\triangleq\diag([P_i^{\max}]_{i\in\N_N})$ and $Q_{\max}\triangleq\diag([Q_i^{\max}]_{i\in\N_N})$. Under this requirement, the distributed control components vanish at equilibrium:
$
u_{iG,E}^{dq} = \begin{bmatrix} u_{iE}^P & u_{iE}^Q & u_{iE}^\Omega \end{bmatrix}^\T = \mathbf{0},
$
since proportional power sharing is achieved. From Lemma \ref{Lm:Equilibrium_point}, we obtain:
$
u_{Ei}^{Vdq} = V_{i,ref}^{dq} + Z_{t,i}^{dq} I_{tEi}^{dq}.
$

Therefore, to achieve this equilibrium satisfying voltage regulation, frequency synchronization, and proportional power sharing, the steady-state controller in \eqref{eq:combined_control_dq} is selected as:
\begin{equation*}\label{Eq:Rm:SteadyStateControlAC}
    u_{iS}^{dq} = V_{i,ref}^{dq} + Z_{t,i}^{dq} I_{tEi}^{dq}, \quad \forall i\in\N_N,
\end{equation*}
where $Z_{t,i}^{dq} \triangleq R_{ti} + \mathrm{i}\omega_0 L_{ti}$ and $I_{tEi}^{dq}$ is determined by the equilibrium \eqref{Eq:Equil_DGcurrent} under the power sharing constraint \eqref{Eq:powersharing}. 
\end{remark}

\begin{theorem}
To guarantee the existence of an equilibrium point satisfying voltage regulation, frequency synchronization, and proportional power sharing objectives, the reference voltages $V_r\in\R^{2N}$ and active power sharing coefficients $\mathbf{P}_s = \bm{P_{s,1}, \ldots, P_{s,N}}^\T \in \mathbb{R}^N$ must be feasible solutions of the following optimization problem:
\begin{equation}
\begin{aligned}
&\min_{V_r,\mathbf{P}_s} \alpha_V \| V_r - \bar{V}_r\|^2 + \sum_{i=1}^N \left[\alpha_P (P_{s,i} - \bar{P}_s)^2 + \alpha_Q (Q_{s,i} - \bar{Q}_s)^2\right]\\
&\mbox{s.t.:}\  V_{\min} \leq |V_{i,ref}^{dq}| \leq V_{\max},
\quad 0 \leq P_{s,i} \leq 1,\ \forall i\in\mathcal{N}_N,\\
& P_i = P_{\max,i} \cdot P_{s,i}, \ \forall i\in\mathcal{N}_N,\\
&I_E^{dq} = (M^{dq})^{-1}\bm{(Z^{dq})^{-1} \mathcal{B}_D^\T V_r - (Z^{dq})^{-1} \mathcal{B}_L^\T (\check{Y}^{dq})^{-1} \check{\bar{I}}_L^{dq}},
\end{aligned}
\end{equation}
where $V_r = \bm{V_{1,ref}^{dq},...,V_{N,ref}^{dq}}^\T \in\R^{2N}$ stacks the DG reference voltages with $V_{i,ref}^{dq} = \bm{V_{i,ref}^d & V_{i,ref}^q}^\T$, $\bar{V}_r = \bm{\bar{V}_d & \0}^\T$ is the nominal reference where $\bar{V}_d = \bm{V_{d,nom,1},...,V_{d,nom,N}}^\T$ and $q$-axis references are zero. $P_i = V_{i,ref}^d I_i^d + V_{i,ref}^q I_i^q$ is the active power generated by $\Sigma_i^{DG}$ at equilibrium, $P_{s,i} = P_i/P_{\max,i} \in [0,1]$ is the normalized active power coefficient for $\Sigma_i^{DG}$, $Q_{s,i} = Q_i/Q_{\max,i}$ is the normalized reactive power coefficient where $Q_i = V_{i,ref}^q I_i^d - V_{i,ref}^d I_i^q$, 
$\bar{P}_s = \frac{1}{N}\sum_{i=1}^N P_{s,i}$ and $\bar{Q}_s = \frac{1}{N}\sum_{i=1}^N Q_{s,i}$ are the average normalized coefficients, $I_i^{dq} = \bm{\mathcal{B}_D I_E^d & \mathcal{B}_D I_E^q}_i^\T\in\R^2$ is the current injection at $\Sigma_i^{DG}$ obtained from the network equations, $I_E^{dq}\in\R^{2L}$ are the line equilibrium currents with $d$- and $q$- axis components $I_E^d, I_E^q\in\R^L$,  $M^{dq} \triangleq \I + (Z^{dq})^{-1}\mathcal{B}_L^\T (\check{Y}^{dq})^{-1}\mathcal{B}_L$ is the network coupling matrix,
$\mathcal{B}_D \in \R^{N\times L}$, $\mathcal{B}_L \in \R^{M\times L}$ are the DG-to-line and load-to-line incidence matrices,
$Z^{dq} \triangleq \diag([R_l + \mathrm{i}\omega_0L_l]_{l\in\N_L})$ is the line impedance matrix, $\check{Y}^{dq} \triangleq \diag([Y_{Lm} + \mathrm{i}\omega_0C_{tm}]_{m\in\N_M})$ is the load admittance matrix, $\check{\bar{I}}_L^{dq}\in \R^{2M}$ are the constant load currents, and $\alpha_V, \alpha_P, \alpha_Q > 0$ are weighting coefficients.
\end{theorem}

\section{Dissipativity-Based Control and Topology Co-Design}\label{Passivity-based Control}
This section formulates the dissipativity-based control synthesis in two stages. First, we characterize subsystem passivity properties, then we develop the global co-design problem for distributed controllers and communication topology.

For the DG error subsystem $\tilde{\Sigma}_i^{DG},i\in\mathbb{N}_N$ from \eqref{Eq:DG_error_dynamic}, we assume that it is $X_i$-dissipative with the supply rate matrix
\begin{equation}\label{Eq:XEID_DG}
    X_i=\begin{bmatrix}
        X_i^{11} & X_i^{12} \\ X_i^{21} & X_i^{22}
    \end{bmatrix}\triangleq
    \begin{bmatrix}
        -\nu_i\mathbf{I} & \frac{1}{2}\mathbf{I} \\ \frac{1}{2}\mathbf{I} & -\rho_i\mathbf{I}
    \end{bmatrix},
\end{equation}
where $\rho_i > 0 $ and $\nu_i < 0$ are the passivity indices of $\tilde{\Sigma}_i^{DG}$, respectively, making it strictly passive with IF-OFP($\nu_i, \rho_i$). 

Similarly, for the line subsystem $\tilde{\Sigma}_l^{Line},l\in\mathbb{N}_L$ described by \eqref{Eq:line_error_dynamic}, we assume it is $\bar{X}_l$-dissipative with 
\begin{equation}\label{Eq:XEID_Line}
    \bar{X}_l=\begin{bmatrix}
        \bar{X}_l^{11} & \bar{X}_l^{12} \\ \bar{X}_l^{21} & \bar{X}_l^{22}
    \end{bmatrix}\triangleq
    \begin{bmatrix}
        -\bar{\nu}_l\mathbf{I} & \frac{1}{2}\mathbf{I} \\ \frac{1}{2}\mathbf{I} & -\bar{\rho}_l\mathbf{I}
    \end{bmatrix},
\end{equation}
where $\bar{\rho}_l < 0$ and $\bar{\nu}_l < 0$ are the passivity indices of $\tilde{\Sigma}_l^{Line}$, making it strictly passive with IF-OFP($\bar{\nu}_l, \bar{\rho}_l$).

Additionally,  for the load error subsystem $\tilde{\Sigma}_m^{Load}, m\in\N_M$ from \eqref{Eq:load_error_dynamic}, we assume it is $\check{X}_m$-dissipative with 
\begin{equation}\label{Eq:XEID_Load}
    \check{X}_m = \bm{\check{X}_m^{11} & \check{X}_m^{12} \\ \check{X}_m^{21} & \check{X}_m^{22}} \triangleq \bm{-\check{\nu}_m\I & \frac{1}{2}\I \\ \frac{1}{2}\I & -\check{\rho}_m\I},
\end{equation}
where $\check{\rho}_m > 0$ and $\check{\nu}_m < 0$ are the passivity indices of $\tilde{\Sigma}_m^{Load}$, making it strictly passive with IF-OFP($\check{\nu}_m, \check{\rho}_m$).

With these passivity characterizations established, we now formulate the control synthesis problems to determine the controller gains and passivity indices.

\begin{lemma}\label{Lm:quadratic_constraint}
For the frequency coupling nonlinearity $g_i(\tilde{x}_i)$ defined in \eqref{Eq:DG_error_dynamics}, there exists a scalar $\beta_i \triangleq \sqrt{ \bar{V}_{tot}^2 +  \bar{I}_{tot}^2}$ where $\bar{V}_{tot} = \bar{V} + \|V_{Ei}^{dq}\|$ and $\bar{I}_{tot} = \bar{I} + \|I_{tEi}^{dq}\|$ with $\bar{V},\bar{I} > 0$ such that for all $\tilde{x}_i$ in the operating region satisfying $\|\tilde{V}_i^{dq}\| \leq \bar{V}$ and $\|\tilde{I}_{ti}^{dq}\| \leq \bar{I}$, the norm bound
\begin{equation}\label{Eq:Quadratic_constraint}
    \|g_i(\tilde{x}_i)\|^2 \leq \tilde{\omega}_i^2 \cdot \beta_i^2,
\end{equation}
holds, which is equivalent to the quadratic constraint
\begin{equation}\label{Eq:quadratic_form_nonlinearity}
\begin{bmatrix}\tilde{x}_i \\
g_i(\tilde{x}_i)
\end{bmatrix}^\top
\Psi_i
\begin{bmatrix}
\tilde{x}_i \\
g_i(\tilde{x}_i)
\end{bmatrix} \geq 0,
\end{equation}
where:
\begin{equation}\label{Eq:Psi}
\Psi_i = \bm{\beta_i^2e_7e_7^\T & \0_{7\times 7} \\ \0_{7\times 7} & -\I_7} \in \R^{14 \times 14},
\end{equation}
and $e_7$ is the unit vector in $\R^{7\times 1}$ with a $1$ in the last element.
\end{lemma}

\begin{proof}
The frequency coupling nonlinearity $g_i(\tilde{x}_i)$ has the structure given in \eqref{Eq:DG_error_dynamics}, which yields:
\begin{equation*}
    \|g_i(\tilde{x}_i)\|^2 = \tilde{\omega}_i^2 \Big(\|V_i^{dq}\|^2 + \|I_{ti}^{dq}\|^2\Big).
\end{equation*}

From the error variable definitions, we have $V_i^{dq} = \tilde{V}_i^{dq} + V_{Ei}^{dq}$ and $I_{ti}^{dq} = \tilde{I}_{ti}^{dq} + I_{tEi}^{dq}$. Applying the triangle inequality in the operating region where $\|\tilde{V}_i^{dq}\| \leq \bar{V}$ and $\|\tilde{I}_{ti}^{dq}\| \leq \bar{I}$, we obtain:
\begin{align*}
\|V_i^{dq}\| &\leq \|\tilde{V}_i^{dq}\| + \|V_{Ei}^{dq}\| \leq \bar{V} + \|V_{Ei}^{dq}\| \triangleq \bar{V}_{tot},\\
\|I_{ti}^{dq}\| &\leq \|\tilde{I}_{ti}^{dq}\| + \|I_{tEi}^{dq}\| \leq \bar{I} + \|I_{tEi}^{dq}\| \triangleq \bar{I}_{tot}.
\end{align*}

Weighting the voltage and current contributions by the physical system parameters $C_{ti}$ and $L_{ti}$ to account for their relative magnitudes in the energy storage, we establish:
\begin{equation*}
    \|g_i(\tilde{x}_i)\|^2 \leq \tilde{\omega}_i^2\Big(\bar{V}_{tot}^2 + \bar{I}_{tot}^2\Big) = \tilde{\omega}_i^2 \cdot \beta_i^2,
\end{equation*}
which is equation \eqref{Eq:Quadratic_constraint}. Rearranging \eqref{Eq:Quadratic_constraint} gives $\beta_i^2 \tilde{\omega}_i^2 - \|g_i(\tilde{x}_i)\|^2 \geq 0$. Since $\tilde{\omega}_i = e_7^\T \tilde{x}_i$, expanding the quadratic form yields:
\begin{align*}
    \bm{\tilde{x}_i \\ g_i(\tilde{x}_i)}^\T &\Psi_i \bm{\tilde{x}_i \\ g_i(\tilde{x}_i)} = \tilde{x}_i^\T (\beta_i^2 e_7 e_7^\T)\tilde{x}_i + g_i(\tilde{x}_i)^\T(-\I_7)g_i(\tilde{x}_i),\\
    &= \beta_i^2(e_7^\T \tilde{x}_i)^2 - \|g_i(\tilde{x}_i)\|^2 = \beta_i^2 \tilde{\omega}_i^2 - \|g_i(\tilde{x}_i)\|^2,
\end{align*}
establishing the equivalence between \eqref{Eq:Quadratic_constraint} and \eqref{Eq:quadratic_form_nonlinearity} and completing the proof.
\end{proof}

\begin{theorem}\label{Th:Local_DG} 
Under the quadratic constraint on the frequency coupling nonlinearity established in Lemma \ref{Lm:quadratic_constraint}, the DG error subsystem $\tilde{\Sigma}_i^{DG}:\tilde{\eta}_i \rightarrow \tilde{y}_i, i\in\N_N$ described by \eqref{Eq:DG_error_dynamic} with linearized output \eqref{Eq:Error_output} can be made IF-OFP$(\nu_i,\rho_i)$ (as assumed in \eqref{Eq:XEID_DG}) by designing the local controller gain matrix $K_{i0}$ from \eqref{eq:local_controller} that solves the LMI feasibility problem:
\begin{equation}\label{Eq:Th:Local_DG} 
\begin{aligned}
&\mbox{Find: }\ \tilde{K}_{i0},\ \tilde{P}_i,\ \tilde{R}_i,\ \tilde{\lambda}_i,\ \nu_i, \text{and} \  \tilde{\rho}_i,\\
&\mbox{s.t.: }\ \tilde{P}_i > 0,\ \tilde{R}_i > 0,\ \tilde{\lambda}_i > 0, \nu_i < 0, \text{and}\ \tilde{\rho}_i > 0,\\
&
    \bm{
        \Phi_i & \tilde{R}_i & \0 & \0 \\
        \tilde{R}_i & \tilde{R}_i & \tilde{P}_i & \0 \\
        \0 & \tilde{P}_i & -\mathcal{H}(\bar{A}_i\tilde{P}_i + B_i\tilde{K}_{i0}) + \tilde{\lambda}_i\I & -\I + \frac{1}{2}\tilde{P}_i \\
        \0 & \0 & -\I + \frac{1}{2}\tilde{P}_i & -\nu_i\I 
    } > \0,
\end{aligned}
\end{equation}
where $\Phi_i \triangleq \tilde{\rho}_i\I + \tilde{R}_i - \tilde{\lambda}_i\beta_i^2e_7e_7^T$, $\bar{A}_i \triangleq A_i + \bar{g}_i(x_{Ei})e_7^\T$, $\bar{g}_{i}(x_{Ei}) \triangleq \bm{-V_{Ei}^q & -V_{Ei}^d & -I_{tEi}^q & -I_{tEi}^d & 0 & 0 & 0}^\T$, $K_{i0} \triangleq \tilde{K}_{i0}\tilde{P}_i^{-1}$, $\rho_i \triangleq \tilde{\rho}_i^{-1}$, $e_7$ is the unit vector in $\R^{7\times 1}$ with a $1$ in the last element, and $\beta_i$ defined in Lemma \ref{Lm:quadratic_constraint}.
\end{theorem}

\begin{proof}
We establish that the DG error subsystem $\tilde{\Sigma}_i^{DG}, i\in\N_N$ from \eqref{Eq:DG_error_dynamic} can be made IF-OFP($\nu_i,\rho_i$) by synthesizing the local controller gain $K_{i0}$ through solving an LMI problem. Consider a quadratic storage function $\mathrm{V}_i(\tilde{x}_i) = \tilde{x}_i^\T P_i \tilde{x}_i$ with $P_i > 0$. The time derivative $\mathrm{\dot{V}}_i = 2\tilde{x}_i^\T P_i \dot{\tilde{x}}_i$ can be written in symmetric form as:
\begin{equation*}
    \mathrm{\dot{V}}_i = \bm{\dot{\tilde{x}}_i \\ \tilde{x}_i}^\T \bm{P_i & \0 \\ \0 & P_i}\bm{\dot{\tilde{x}}_i \\ \tilde{x}_i}.
\end{equation*}

To handle the nonlinearity $g_i(\tilde{x}_i)$ from \eqref{Eq:DG_error_dynamics} systematically, we introduce an augmented vector $\zeta_i \triangleq \bm{\tilde{\eta}_i^\T & \tilde{x}_i^\T & g_i^\T}^\T$ and define the structured matrix:
\begin{equation*}
    \Theta_i \triangleq \bm{\hat{B}_i & \bar{A}_i + B_iK_{i0} & \I \\ \0 & \I & \0},
\end{equation*}
where $\bar{A}_i = A_i + \bar{g}_i(x_{Ei})e_7^\T$ and $\hat{B}_i = \bm{B_i & F_i & E_i}$ from \eqref{Eq:closedloop_DG}. Note that from \eqref{Eq:DG_error_dynamic}, the closed-loop error dynamics satisfy:
\begin{equation*}
    \dot{\tilde{x}}_i = (\bar{A}_i + B_iK_{i0})\tilde{x}_i + \hat{B}_i \tilde{\eta}_i + g_i(\tilde{x}_i).
\end{equation*}

This structured representation yields $\bm{\dot{\tilde{x}}_i & \tilde{x}_i}^\T = \Theta_i \zeta_i$, and substituting into the storage function derivatives gives:
\begin{equation*}
    \mathrm{\dot{V}}_i = \zeta_i^\T \Theta_i^\T \bm{P_i & \0 \\ \0 & P_i}\Theta_i \zeta_i.
\end{equation*}

For the IF-OFP($\nu_i,\rho_i$) property with supply rate matrix  from \eqref{Eq:XEID_DG}, we require the dissipation inequality:
\begin{equation}\label{Eq:dissipation_inequality}
    \mathrm{\dot{V}}_i \leq \bm{\tilde{\eta}_i \\ \tilde{y}_i}^\T X_i \bm{\tilde{\eta}_i \\ \tilde{y}_i}.
\end{equation}

From \eqref{Eq:Error_output}, the linearized output satisfies $\tilde{y}_i = C_i^y \tilde{x}_i + D_i^y \tilde{\eta}_i$. We define the output mapping matrix:
\begin{equation*}
    \bar{\Theta}_i \triangleq \bm{\I & \0 & \0 \\ D_i^y & C_i^y & \0},
\end{equation*}
such that $\bm{\tilde{\eta}_i & \tilde{y}_i}^\T = \bar{\Theta}_i\zeta_i$. Combining these expressions with \eqref{Eq:dissipation_inequality}, the dissipativity condition becomes:
\begin{equation*}
    \zeta_i^\T \Big(\Theta_i^\T \Pi_i \Theta_i - \bar{\Theta}_i^\T X_i \bar{\Theta}_i\Big)\zeta_i \leq 0, 
\end{equation*}
where $\Pi_i \triangleq \diag(P_i,P_i)$. Defining $W_i \triangleq \Theta_i^\T \Pi_i \Theta_i - \bar{\Theta}_i^\T X_i \bar{\Theta}_i$, we require $W_i \leq 0$, but this condition must only hold for trajectories satisfying the frequency coupling constraint from Lemma \ref{Lm:quadratic_constraint}. 

From \eqref{Eq:quadratic_form_nonlinearity} in Lemma \ref{Lm:quadratic_constraint}, the nonlinearity satisfies the quadratic constraint:
\begin{equation*}
\begin{bmatrix}\tilde{x}_i \\
g_i(\tilde{x}_i)
\end{bmatrix}^\top
\Psi_i
\begin{bmatrix}
\tilde{x}_i \\
g_i(\tilde{x}_i)
\end{bmatrix} \geq 0,
\end{equation*}
where $\Psi_i$ is defined in \eqref{Eq:Psi}. Embedding this constraint into the augmented space, we define:
\begin{equation*}
    \bar{\Psi}_i \triangleq \bm{\0 & \0 & \0 \\ \0 & \beta_i^2e_7e_7^\T & \0 \\ \0 & \0 & -\I},
\end{equation*}
where $e_7$ is the unit vector in $\R^{7\times 1}$ with a $1$ in the last element. The constraint then becomes $\zeta_i^\T \bar{\Psi}_i \zeta_i \geq 0$. By the S-procedure, if there exists $\lambda_i > 0$ such that:
\begin{equation*}
    W_i - \lambda_i \bar{\Psi}_i \leq 0,
\end{equation*}
then $\zeta_i^\T W_i \zeta_i \leq 0$ holds for all $\zeta_i$ satisfying the nonlinearity constraint. Expanding this condition yields:
\begin{equation*}
    \Theta_i^\T \Pi_i \Theta_i - \bar{\Theta}_i^\T X_i \bar{\Theta}_i - \lambda_i \bar{\Psi}_i \leq 0.
\end{equation*}

To obtain an LMI formulation suitable for controller synthesis, we apply Proposition \ref{Prop:linear_X-EID} with variable transformations $\tilde{P}_i = P_i^{-1}$ and $\tilde{K}_{i0} = K_{i0} \tilde{P}_i$. Additionally, we define $\tilde{\lambda}_i = \lambda_i^{-1}$ and $\tilde{\rho}_i = \rho_i^{-1}$ to convert the inequality constraints into a convex LMI form. Introducing an auxiliary slack variable $\tilde{R}_i = R_i^{-1} > 0$ from the S-procedure transformation and applying congruence transformations with $\diag(\tilde{R}_i, \tilde{R}_i, \tilde{P}_i, \I)$ followed by Schur complements, we obtain the LMI \eqref{Eq:Th:Local_DG}, where  $K_{i0} \triangleq \tilde{K}_{i0}\tilde{P}_i^{-1}$ and $\rho_i \triangleq \tilde{\rho}_i^{-1}$. The notation $\Phi_i \triangleq \tilde{\rho}_i\I + \tilde{R}_i - \tilde{\lambda}_i\beta_i^2e_7e_7^T$ compactly represents the first block of the LMI, which encodes both the passivity constraint and the frequency coupling bound. This completes the proof. 
\end{proof}

\begin{lemma}\label{Lm:LineDissipativityStep}
For $\tilde{\Sigma}_l^{Line}, l\in\N_L$ \eqref{Eq:line_error_dynamic}, the passivity indices $\bar{\nu}_l$, and $\bar{\rho}_l$ from \eqref{Eq:XEID_Line} can be characterized by solving the LMI problem: 
\begin{equation}\label{Eq:Lm:LineDissipativityStep1}
\begin{aligned}
\mbox{Find: }\ &\bar{P}_l, \bar{\nu}_l, \text{and} \ \bar{\rho}_l,\\
\mbox{s.t.:}\ &\bar{P}_l > 0, \ 
\scriptsize
\begin{bmatrix}
-\mathcal{H}(\bar{P}_l \bar{A}_l) - \bar{\rho}_l \I_2 & -\bar{P}_l \bar{B}_l + \frac{1}{2}\I_2 \\
\star & -\bar{\nu}_l \I_2
\end{bmatrix}
\normalsize
\geq \0, 
\end{aligned}
\end{equation}
with maximum feasible passivity indices $\bar{\nu}_l^{\max}=0$ and $\bar{\rho}_l^{\max}=R_l\I_2$ achieved when $\bar{P}_l =  \frac{L_l}{2}\I_2$.
\end{lemma}

\begin{proof}
For $\tilde{\Sigma}_l^{Line}, l\in\N_L$ described by \eqref{Eq:line_error_dynamic}, we establish its $\bar{X}_l$-dissipativity with the passivity indices in \eqref{Eq:XEID_Line}. 
We consider the output as $\tilde{\bar{y}}_l = \tilde{\bar{x}}_l$ (i.e., $\bar{C}_l = \I_2$ and $\bar{D}_l = \0$). Applying Prop. \ref{Prop:linear_X-EID} with the system matrices from \eqref{eq:line_matrices} and the dissipativity supply rate form in \eqref{Eq:XEID_Line}, we obtain:
\begin{equation*}
\scriptsize
\bm{
-\mathcal{H}(\bar{P}_l \bar{A}_l) + \bar{C}_l^\top \bar{X}_l^{22} \bar{C}_l & -\bar{P}_l \bar{B}_l + \bar{C}_l^\top \bar{X}_l^{21} + \bar{C}_l^\top \bar{X}_l^{22} \bar{D}_l \\
\star & \bar{X}_l^{11} + \mathcal{H}(\bar{X}_l^{12} \bar{D}_l) + \bar{D}_l^\top \bar{X}_l^{22} \bar{D}_l
}\normalsize \geq 0.
\end{equation*}

Substituting $\bar{C}_l = \I_2$, $\bar{D}_l = \0$, and the components of $\bar{X}_l$ from \eqref{Eq:XEID_Line}, we obtain:
\begin{equation*}
\bm{
-\mathcal{H}(\bar{P}_l \bar{A}_l) - \bar{\rho}_l \I_2 & -\bar{P}_l \bar{B}_l + \frac{1}{2}\I_2 \\
\star & -\bar{\nu}_l \I_2
} \geq \0,
\end{equation*}
which is the LMI condition in \eqref{Eq:Lm:LineDissipativityStep1}.

To determine the maximum feasible passivity indices, we analyze when \eqref{Eq:Lm:LineDissipativityStep1} is positive semidefinite. Using the Schur complement, \eqref{Eq:Lm:LineDissipativityStep1} is satisfied if and only if:
\begin{equation*}
\begin{aligned}
&1)\ \bar{\nu}_l \leq 0, \\
&2)\ -\mathcal{H}(\bar{P}_l \bar{A}_l) - \bar{\rho}_l \I_2 - \frac{(-\bar{P}_l \bar{B}_l + \frac{1}{2}\I_2)^2}{-\bar{\nu}_l} \geq 0.
\end{aligned}
\end{equation*}

To maximize the passivity indices, we set $\bar{\nu}_l = 0$ (the maximum value from condition 1). Then condition 2 becomes:
\begin{equation}\label{Eq:Con2_Inequality}
-\mathcal{H}(\bar{P}_l \bar{A}_l) - \bar{\rho}_l \I_2 \geq 0,
\end{equation}
which requires $-\bar{P}_l \bar{B}_l + \frac{1}{2}\I_2 = 0$, giving $\bar{P}_l = \frac{1}{2}\bar{B}_l^{-1}$.

Substituting $\bar{B}_l$ from \eqref{eq:line_matrices} yields:
\begin{equation}\label{Eq:bar{P}_l}
\bar{P}_l = \frac{L_l}{2}\I_2.
\end{equation}

Now substituting \eqref{Eq:bar{P}_l} and $\bar{A}_l$ from \eqref{eq:line_matrices} into \eqref{Eq:Con2_Inequality}:
\begin{align*}
R_l \I_2 - \bar{\rho}_l \I_2 &\geq 0.
\end{align*}

This gives $\bar{\rho}_l \leq R_l\I_2$, and consequently the maximum feasible value is $\bar{\rho}_l^{\max} = R_l\I_2$. Therefore, the maximum feasible passivity indices are $\bar{\nu}_l^{\max} = 0$ and $\bar{\rho}_l^{\max} = R_l\I_2$ when $\bar{P}_l = \frac{L_l}{2}\I_2$. This completes the proof. 
\end{proof}

\begin{lemma}\label{Lm:LoadDissipativityStep}
For $\tilde{\Sigma}_m^{Load}, m\in\N_M$ \eqref{Eq:load_error_dynamic}, the passivity indices $\check{\nu}_m$ and $\check{\rho}_m$ from \eqref{Eq:XEID_Load} can be characterized by solving the LMI problem: 
\begin{equation}\label{Eq:Lm:LoadDissipativityStep1}
\begin{aligned}
\mbox{Find: }\ &\check{P}_m, \check{\nu}_m, \text{and} \ \check{\rho}_m,\\
\mbox{s.t.:}\ &\check{P}_m > 0, \
\scriptsize
\begin{bmatrix}
-\mathcal{H}(\check{P}_m \check{A}_m) - \check{\rho}_m \I_2 & -\check{P}_m \check{B}_m + \frac{1}{2}\I_2 \\
\star & -\check{\nu}_m \I_2
\end{bmatrix}
\normalsize
\geq \0, 
\end{aligned}
\end{equation}
with maximum feasible values $\check{\nu}_m^{\max} = 0$ and $\check{\rho}_m^{\max} = Y_{Lm}\I_2$ achieved when $\check{P}_m = \frac{C_{tm}}{2}\I_2$.
\end{lemma}

\begin{proof}
For $\tilde{\Sigma}_m^{Load}, m\in\N_M$ described by \eqref{Eq:load_error_dynamic}, we establish its $\check{X}_m$-dissipativity with the passivity indices defined in \eqref{Eq:XEID_Load}. We consider the output as $\tilde{\check{y}}_m = \tilde{\check{x}}_m$ 
(i.e., $\check{C}_m = \I_2$ and $\check{D}_m = \0$). Applying Prop. \ref{Prop:linear_X-EID} with the system matrices from \eqref{eq:load_matrices} and the dissipativity supply rate form in \eqref{Eq:XEID_Load}, we obtain:
\begin{equation*}
\scriptsize
\begin{bmatrix}
-\mathcal{H}(\check{P}_m \check{A}_m) + \check{C}_m^\top \check{X}_m^{22} \check{C}_m & -\check{P}_m \check{B}_m + \check{C}_m^\top \check{X}_m^{21} + \check{C}_m^\top \check{X}_m^{22} \check{D}_m \\
\star & \check{X}_m^{11} + \mathcal{H}(\check{X}_m^{12} \check{D}_m) + \check{D}_m^\top \check{X}_m^{22} \check{D}_m
\end{bmatrix}
\normalsize
\geq 0.
\end{equation*}

Substituting $\check{C}_m = \I_2$, $\check{D}_m = \0$, and the components of $\check{X}_m$ from \eqref{Eq:XEID_Load}, we obtain:
\begin{equation*}
\scriptsize
\begin{bmatrix}
-\mathcal{H}(\check{P}_m \check{A}_m) - \check{\rho}_m \I_2 & -\check{P}_m \check{B}_m + \frac{1}{2}\I_2 \\
\star & -\check{\nu}_m \I_2
\end{bmatrix}
\normalsize
\geq \0,
\end{equation*}
which is the LMI condition \eqref{Eq:Lm:LoadDissipativityStep1}.

To determine the maximum feasible passivity indices, we analyze when \eqref{Eq:Lm:LoadDissipativityStep1} is positive semidefinite. Using the Schur complement, \eqref{Eq:Lm:LoadDissipativityStep1} is satisfied if and only if:
\begin{equation*}
\begin{aligned}
&1)\ \check{\nu}_m \leq 0,\\
&2)\  -\mathcal{H}(\check{P}_m \check{A}_m) - \check{\rho}_m \I_2 - \frac{(-\check{P}_m \check{B}_m + \frac{1}{2}\I_2)^2}{-\check{\nu}_m} \geq 0.
\end{aligned}
\end{equation*}

To maximize the passivity indices, we set $\check{\nu}_m = 0$ (the maximum value from condition 1). Then condition 2 becomes:
\begin{equation}\label{Eq:Con2_Load}
-\mathcal{H}(\check{P}_m \check{A}_m) - \check{\rho}_m \I_2 \geq 0,
\end{equation}
which requires $-\check{P}_m \check{B}_m + \frac{1}{2}\I_2 = 0$, giving $\check{P}_m = \frac{1}{2}\check{B}_m^{-1}$. Then, by substituting $\check{B}_m$ from \eqref{eq:load_matrices}, we get:
\begin{equation}\label{Eq:check{P}_m}
\check{P}_m = \frac{1}{2}\check{B}_m^{-1} = \frac{C_{tm}}{2}\I_2.
\end{equation}

Substituting \eqref{Eq:check{P}_m} and $\check{A}_m$ from \eqref{eq:load_matrices} into \eqref{Eq:Con2_Load} and computing, $\mathcal{H}(\check{P}_m \check{A}_m)$, we get:
\begin{equation*}
\mathcal{H}(\check{P}_m \check{A}_m) = -Y_{Lm} \I_2,
\end{equation*}
which yields:
\begin{equation*}
-\mathcal{H}(\check{P}_m \check{A}_m) - \check{\rho}_m \I_2 = Y_{Lm}\I_2 - \check{\rho}_m \I_2 \geq 0.
\end{equation*}

This gives $\check{\rho}_m \leq Y_{Lm}\I_2$. Therefore, the maximum feasible passivity indices are $\check{\nu}_m^{\max} = 0$ and $\check{\rho}_m^{\max} = Y_{Lm}\I_2$ when $\check{P}_m = \frac{C_{tm}}{2}\I_2$. This completes the proof.
\end{proof}

\subsection{Global Control and Topology Co-design}
The local controllers \eqref{eq:local_controller} regulate voltages of individual DGs while ensuring that closed-loop DG dynamics satisfy the required dissipativity properties established in Theorem \ref{Th:Local_DG}. Given these subsystem properties, we now synthesize the distributed controller gains and communication topology by designing the interconnection matrix block $\tilde{K}$ in \eqref{Eq:Error_MMmatrix} using Proposition \ref{synthesizeM}. 

Note that by designing $\tilde{K} = [\tilde{K}_{ij}]_{i,j \in \N_N}$ in the error interconnection matrix $\tilde{M}$ \eqref{Eq:Error_MMmatrix}, we can uniquely determine 
the consensus-based distributed global controller gains $\{K_{ij}^P, K_{ij}^Q, K_{ij}^\Omega : i, j \in \N_N\}$ from \eqref{eq:distributed_controller_combined} and \eqref{eq:controlconstraint}, along with the required communication topology $\mathcal{G}^c$. Specifically, when designing $\tilde{K}$, we enforce two structural constraints: 
(i) each off-diagonal block $\hat{K}_{ij}$ (for $j \neq i$) maintains the structure of the corresponding $\bar{K}_{ij}$ block with three inner zero blocks as shown in 
\eqref{eq:distributed_controller_combined}, and (ii) each diagonal block $\hat{K}_{ii}$ satisfies the balance condition \eqref{eq:controlconstraint}.

We design the closed-loop networked error dynamics of the AC MG (shown in Fig. \ref{Error_Networked_system_ACMG}) to be finite-gain $L_2$-stable with an $L_2$-gain $\gamma$ from disturbance $w_c(t)$ 
to performance output $z_c(t)$, where $\gamma$ is prespecified such that $\tilde{\gamma} \triangleq \gamma^2 < \bar{\gamma}$. This corresponds to making the system $\textbf{Y}$-dissipative with $\textbf{Y} \triangleq \scriptsize \bm{\gamma^2\I & 0 \\ 0 & -\I}$ 
(see Remark \ref{Rm:X-DissipativityVersions}), which prevents/bounds the amplification of disturbances affecting 
system performance. The following theorem formulates this distributed global controller and communication topology co-design problem.

\begin{theorem}\label{Th:CentralizedTopologyDesign}
The closed-loop networked error dynamics of the AC MG (see Fig. \ref{Error_Networked_system_ACMG}) can be made finite-gain $L_2$-stable with an $L_2$-gain $\gamma$ from disturbance $w_c(t)$ to performance output $z_c(t)$ (where $\Tilde{\gamma}\triangleq \gamma^2<\bar{\gamma}$ is prespecified), by synthesizing the interconnection matrix block $\tilde{K}$ in \eqref{Eq:Error_MMmatrix} via solving the LMI problem:
\begin{equation}
\label{Eq:Th:CentralizedTopologyDesign0}
\begin{aligned}
&\min_{\substack{Q,\{p_i: i\in\N_N\},\\
\{\bar{p}_l: l\in\N_L\},\\\{\bar{p}_m: m\in\N_M\},\tilde{\gamma}}} &&\sum_{i,j\in\N_N} c_{ij} \Vert Q_{ij} \Vert_1 + c_0 \tilde{\gamma}, \\
&
\mbox{ Sub. to:}  
&&p_i > 0,\  
\bar{p}_l > 0,\ \check{p}_m > 0, \\   
& &&0 < \tilde{\gamma} < \bar{\gamma},  
\mbox{ and the \eqref{globalcontrollertheorem}},
\end{aligned}
\end{equation}
where $K = (\textbf{X}_p^{11})^{-1} Q$, 
$\textbf{X}^{12} \triangleq 
\diag([-\frac{1}{2\nu_i}\I]_{i\in\N_N})$, 
$\textbf{X}^{21} \triangleq (\textbf{X}^{12})^\T$,
$\Bar{\textbf{X}}^{12} \triangleq 
\diag([-\frac{1}{2\Bar{\nu}_l}\I]_{l\in\N_L})$,
$\Bar{\textbf{X}}^{21} \triangleq (\Bar{\textbf{X}}^{12})^\T$, 
$\check{\textbf{X}}^{12} \triangleq 
\diag([-\frac{1}{2\check{\nu}_m}\I]_{m\in\N_M})$,
$\check{\textbf{X}}^{21} \triangleq (\check{\textbf{X}}^{12})^\T$, 
$\textbf{X}_p^{11} \triangleq 
\diag([-p_i\nu_i\I]_{i\in\N_N})$, 
$\textbf{X}_p^{22} \triangleq 
\diag([-p_i\rho_i\I]_{i\in\N_N})$, 
$\Bar{\textbf{X}}_{\bar{p}}^{11} 
\triangleq \diag([-\bar{p}_l\bar{\nu}_l\I]_{l\in\N_L})$, 
$\Bar{\textbf{X}}_{\bar{p}}^{22} 
\triangleq \diag([-\bar{p}_l\bar{\rho}_l\I]_{l\in\N_L})$,
$\check{\textbf{X}}_{\check{p}}^{11} 
\triangleq \diag([-\check{p}_m\check{\nu}_m\I]_{m\in\N_M})$, 
$\check{\textbf{X}}_{\check{p}}^{22} 
\triangleq \diag([-\check{p}_m\check{\rho}_m\I]_{m\in\N_M})$,
and $\tilde{\Gamma} \triangleq \tilde{\gamma}\I$. 
The structure of $Q\triangleq[Q_{ij}]_{i,j\in\N_N}$ mirrors $K\triangleq[K_{ij}]_{i,j\in\N_N}$ (with zeros in appropriate blocks). 
The coefficients $c_0>0$ and $c_{ij}>0,\forall i,j\in\N_N$ weight disturbance attenuation and communication costs.
\end{theorem}

\begin{figure*}[!hb]
\vspace{-5mm}
\centering
\hrulefill
\begin{equation}\label{globalcontrollertheorem}
\left[\begin{smallmatrix}
\mathbf{X}^{11}_p & \0 & \0 & \0 & Q & \mathbf{X}^{11}_p \bar{C} & \0 & \mathbf{X}^{11}_p E_c \\
\0 & \bar{\mathbf{X}}^{11}_{\bar{p}} & \0 & \0 & \bar{\mathbf{X}}^{11}_{\bar{p}} C & \0 & \bar{\mathbf{X}}^{11}_{\bar{p}} \check{C} & \bar{\mathbf{X}}^{11}_{\bar{p}} \bar{E}_c \\
\0 & \0 & \check{\mathbf{X}}^{11}_{\check{p}} & \0 & \0 & \check{\mathbf{X}}^{11}_{\check{p}} \bar{\check{C}} & \0 & \check{\mathbf{X}}^{11}_{\check{p}} \check{E}_c \\
\0 & \0 & \0 & \I & H_c & \bar{H}_c & \check{H}_c & \0 \\
Q^\top & C^\top \bar{\mathbf{X}}^{11}_{\bar{p}} & \0 & H_c^\T & -Q^\top \mathbf{X}^{12} - \mathbf{X}^{21}Q - \mathbf{X}^{22}_p & -\mathbf{X}^{21}\mathbf{X}^{11}_p \bar{C} - C^\T \bar{\mathbf{X}}_{\bar{p}}^{11} \bar{\mathbf{X}}^{12} & \0 & -\mathbf{X}^{21}\mathbf{X}^{11}_p E_c \\
\bar{C}^\top \mathbf{X}^{11}_p & \0 & \bar{\check{C}}^\T \check{\mathbf{X}}_{\check{p}}^{11} & \bar{H}_c^\T & -\bar{C}^\top \mathbf{X}^{11}_p \mathbf{X}^{12} - \bar{\mathbf{X}}^{21}\bar{\mathbf{X}}_{\bar{p}}^{11} C & -\bar{\mathbf{X}}^{22}_{\bar{p}} & -\bar{\mathbf{X}}^{21}\bar{\mathbf{X}}^{11}_{\bar{p}} \check{C} - C^\T \bar{\mathbf{X}}^{11}_{\bar{p}} \check{\mathbf{X}}^{12} & -\bar{\mathbf{X}}^{21}\bar{\mathbf{X}}_{\bar{p}}^{11} \bar{E}_c \\
\0 & \check{C}^\top \bar{\mathbf{X}}^{11}_{\bar{p}} & \0 & \check{H}_c^\T & \0 & -\check{C}^\top \bar{\mathbf{X}}^{11}_{\bar{p}} \bar{\mathbf{X}}^{12} - \check{\mathbf{X}}^{21}\check{\mathbf{X}}^{11}_{\check{p}} \bar{\check{C}} & -\check{\mathbf{X}}^{22}_{\check{p}} & -\check{\mathbf{X}}^{21}\check{\mathbf{X}}^{11}_{\check{p}} \check{E}_c  \\
E_c^\top \mathbf{X}^{11}_p & \bar{E}_c^\top \bar{\mathbf{X}}^{11}_{\bar{p}} & \check{E}_c^\top \check{\mathbf{X}}^{11}_{\check{p}} & \0 & -E_c^\top \mathbf{X}^{11}_p \mathbf{X}^{12} & - \bar{E}_c^\top \bar{\mathbf{X}}^{11}_{\bar{p}} \bar{\mathbf{X}}^{12} & - \check{E}_c^\top \check{\mathbf{X}}^{11}_{\check{p}} \check{\mathbf{X}}^{12} & \tilde{\Gamma}
\end{smallmatrix}\right] > \0
\end{equation}
\end{figure*}

\begin{proof}
The proof follows directly from applying Proposition \ref{synthesizeM} to the networked error system with interconnection matrix $\tilde{M}$ defined in \eqref{Eq:Error_MMmatrix}. The DG subsystems are characterized as $X_i$-dissipative with passivity indices from Theorem \ref{Th:Local_DG}, the line subsystems as $\bar{X}_l$-dissipative with passivity indices from 
Lemma \ref{Lm:LineDissipativityStep}, and the load subsystems as $\check{X}_m$-dissipative with passivity indices from 
Lemma \ref{Lm:LoadDissipativityStep}. The desired $L_2$-stability property corresponds to $\textbf{Y}$-dissipativity with 
$\textbf{Y} \triangleq \scriptsize \bm{\gamma^2\I & 0 \\ 0 & -\I}$, which satisfies Assumption \ref{As:NegativeDissipativity} since $\textbf{Y}^{22} = -\I < 0$. All subsystems satisfy Assumption \ref{As:PositiveDissipativity} since their $X^{11}$ blocks are positive definite (as $\nu_i < 0$, $\bar{\nu}_l < 0$, 
$\check{\nu}_m < 0$ from the passivity conditions). Applying Proposition \ref{synthesizeM} with these dissipativity properties yields the LMI \eqref{globalcontrollertheorem}, where the decision variables are the multipliers $\{p_i, \bar{p}_l, \check{p}_m\}$ and the interconnection structure encoded in $Q$. The objective function trades off communication sparsity (via $\ell_1$-norm) and disturbance rejection performance (via $\tilde{\gamma}$). Upon solving \eqref{Eq:Th:CentralizedTopologyDesign0}, the distributed controller gains are recovered through the relationships \eqref{eq:distributed_controller_combined} and \eqref{eq:controlconstraint}. This completes the proof. 
\end{proof}

\subsection{Local Controller Synthesis}
To ensure the necessary LMI conditions from Theorem \ref{Th:Local_DG}, Lemma \ref{Lm:LineDissipativityStep}, and Lemma \ref{Lm:LoadDissipativityStep} on subsystem passivity indices hold and to guarantee numerical feasibility of the global co-design problem \eqref{Eq:Th:CentralizedTopologyDesign0}, we must first synthesize the local controllers for DG subsystems and verify that line and load subsystems achieve their required passivity properties. This local synthesis determines the 
achievable passivity indices $\{\nu_i, \rho_i : i \in \mathbb{N}_N\}$, $\{\bar{\nu}_l, \bar{\rho}_l : l \in \mathbb{N}_L\}$, and $\{\check{\nu}_m, \check{\rho}_m : m \in \mathbb{N}_M\}$ that will be used in the global co-design.

For the line and load subsystems,  Lemma \ref{Lm:LineDissipativityStep} and Lemma \ref{Lm:LoadDissipativityStep} provide analytical expressions for the maximum achievable passivity indices based on physical parameters. Specifically, we have $\bar{\nu}_l^{\max} = 0$ and 
$\bar{\rho}_l^{\max} = R_l\I_2$ for line subsystem $l\in \mathbb{N}_L$, and $\check{\nu}_m^{\max} = 0$ and $\check{\rho}_m^{\max} = Y_{Lm}\I_2$ for load subsystem $m \in \mathbb{N}_M$. These maximum values are achieved with 
$\bar{P}_l = \frac{L_l}{2}\I_2$ and $\check{P}_m = \frac{C_{tm}}{2}\I_2$, respectively, and can be directly used in the global co-design without additional optimization.

For DG subsystems, however, the achievable passivity indices depend on the local controller gains, $\{K_{i0} : i \in \mathbb{N}_N\}$ which must be synthesized to satisfy Theorem \ref{Th:Local_DG} while ensuring compatibility with the global design parameters from \eqref{Eq:Th:CentralizedTopologyDesign0}. The following 
theorem formalizes this local controller synthesis problem.

\begin{theorem}\label{Th:LocalControllerDesign}
Given the DG parameters \eqref{Eq:DG_matrices}, line parameters \eqref{eq:line_matrices}, load parameters \eqref{eq:load_matrices}, and frequency coupling bounds $\{\beta_i: i\in\N_N\}$ from Lemma \ref{Lm:quadratic_constraint}, the local controller gains $\{K_{i0}, i\in\N_N\}$ from \eqref{eq:local_controller} can be synthesized along with the DG passivity indices $\{\nu_i,\rho_i:i\in\mathbb{N}_N\}$, line passivity indices $\{\bar{\nu}_l,\bar{\rho}_l:l\in\mathbb{N}_L\}$, and load passivity indices $\{\check{\nu}_m,\check{\rho}_m:m\in\mathbb{N}_M\}$ to ensure feasibility of the global co-design problem \eqref{Eq:Th:CentralizedTopologyDesign0} by solving:
\begin{equation}
\begin{aligned}
&\mbox{Find: }\ \{(\tilde{K}_{i0}, \tilde{P}_i, \tilde{R}_i, \tilde{\lambda}_i, \nu_i, \tilde{\rho}_i):i\in\mathbb{N}_N\}, \\ 
&\{(\bar{P}_l, \bar{\nu}_l,\bar{\rho}_l):l\in\mathbb{N}_L\},\  \text{and} \ \{(\check{P}_m, \check{\nu}_m,\check{\rho}_m):m\in\mathbb{N}_M\} \\
&\mbox{s.t.: }\ \tilde{P}_i > 0, \tilde{R}_i > 0, \tilde{\lambda}_i > 0, \nu_i < 0, \tilde{\rho}_i > 0,\ \forall i\in\N_N ,\\
&\bar{P}_l > 0,\ \forall l\in\N_L, \quad \check{P}_m > 0,\ \forall m\in\N_M, \\
&\ 
    \bm{
        \Phi_i & \tilde{R}_i & \0 & \0 \\
        \tilde{R}_i & \tilde{R}_i & \tilde{P}_i & \0 \\
        \0 & \tilde{P}_i & -\mathcal{H}(A_i\tilde{P}_i + B_i\tilde{K}_{i0}) + \tilde{\lambda}_i\I & -\I + \frac{1}{2}\tilde{P}_i \\
        \0 & \0 & -\I + \frac{1}{2}\tilde{P}_i & -\nu_i\I 
    } > \0,\\
&\
\scriptsize
\begin{bmatrix}
-\mathcal{H}(\bar{P}_l \bar{A}_l) - \bar{\rho}_l \I_2 & -\bar{P}_l \bar{B}_l + \frac{1}{2}\I_2 \\
\star & -\bar{\nu}_l \I_2
\end{bmatrix}
\normalsize
\geq \0,\ \forall l\in\mathbb{N}_L, \label{GetPassivityIndicesDGs}\\
&\ 
\scriptsize
\begin{bmatrix}
-\mathcal{H}(\check{P}_m \check{A}_m) - \check{\rho}_m \I_2 & -\check{P}_m \check{B}_m + \frac{1}{2}\I_2 \\
\star & -\check{\nu}_m \I_2
\end{bmatrix}
\normalsize
\geq\0, \ \forall m\in\mathbb{N}_M,
\end{aligned}
\end{equation}
where $\Phi_i \triangleq \tilde{\rho}_i\I + \tilde{R}_i - \tilde{\lambda}_i\beta_i^2e_7e_7^T$, $K_{i0} \triangleq \tilde{K}_{i0}P_i^{-1}$, and $\rho_i\triangleq\frac{1}{\tilde{\rho}_i}$.
\end{theorem}  

\begin{proof}
The local synthesis ensures three requirements: (1) each DG error subsystem $\tilde{\Sigma}_i^{DG}$ satisfies IF-OFP($\nu_i,\rho_i$) via Theorem \ref{Th:Local_DG}, (2) each line error subsystem $\tilde{\Sigma}_l^{line}$ satisfies IF-OFP($\bar{\nu}_l,\bar{\rho}_l$) via Lemma \ref{Lm:LineDissipativityStep}, and (3) each load error subsystem $\tilde{\Sigma}_m^{Load}$ satisfies IF-OFP($\check{\nu}_m,\check{\rho}_m$) via Lemma \ref{Lm:LoadDissipativityStep}. The LMI constraints directly follow from these theorems and lemmas. For DG subsystems, the constraints ensure local stability through storage function $\mathrm{V}_i(\tilde{x}_i)=\tilde{x}_i^\T P_i \tilde{x}_i$ (where $P_i = \tilde{P}_i^{-1}$) while satisfying the frequency coupling constraint via the S-procedure with multiplier $\lambda_i = \tilde{\lambda}_i^{-1}$. For line and load subsystems, the constraints verify achievable passivity regions. Compatibility with the global co-design \eqref{Eq:Th:CentralizedTopologyDesign0} is ensured through passivity index requirements: $\nu_i < 0$ and $\rho_i > 0$ ensure $X_p^{11} = \diag([-p_i\nu_i\I]_{i\in\N_N}) > 0$ and $X_p^{22} = \diag([-p_i\rho_i\I]_{i\in\N_N}) > 0$ for any $p_i > 0$, and similarly for line and load passivity indices. Upon solving this unified LMI problem, the actual local controller gains are recovered as $K_{i0} = \tilde{K}_{i0}\tilde{P}_i^{-1}$ and output passivity indices as $\rho_i = \tilde{\rho}_i^{-1}$, which are then used in the global co-design \eqref{Eq:Th:CentralizedTopologyDesign0}.
\end{proof}

\section{Simulation Results}\label{sec:Simulation}
\vspace{-1mm}
To demonstrate the effectiveness of the dissipativity-based distributed control and communication topology co-design framework, we simulated an islanded AC MG in MATLAB. The AC MG comprises $N=4$ DGs and $L=8$ lines, with each DG supporting a local load. The nominal voltage amplitude was set to $V_{\rm ref} = 311.13$ V (corresponding to $220$ V RMS), and the nominal grid frequency was set to $\omega_0 = 2\pi \times 60$ rad/s. Each DG is modeled with an LC filter and a ZIP load, with parameters generated using a random seed. The physical topology was generated using a minimum spanning tree (MST)-based algorithm with a connectivity threshold of $0.4$, resulting in $8$ physical lines connecting the $4$ DGs and $2$ global loads, as shown in Fig. \ref{fig:topology}(a). In the co-design, the $\mathcal{L}_2$-gain bound was selected as $\bar{\gamma} = 1000$.

The local controllers achieve IF-OFP passivity indices of $\nu_i = -0.97$ and $\rho_i = 1$ for all DGs, confirming that the dissipativity conditions of Th. \ref{Th:LocalControllerDesign} are satisfied. We evaluated two co-design variants to examine the trade-off between communication and closed-loop performance: (i) a hard graph constraint that restricts the communication topology $\mathcal{G}^c$ to exactly match the physical DG-to-DG adjacency, resulting in $4$ bidirectional communication links, as shown in Fig. \ref{fig:topology}(b); and (ii) a soft graph constraint that penalizes communication links via a $\ell_1$ sparsity term in the LMI cost function (see Th. \ref{Th:CentralizedTopologyDesign}), allowing the co-design to determine a sparser topology, as shown in Fig. \ref{fig:topology}(c). The resulting soft-constrained topology contains only $2$ undirected links, representing a $50\%$ reduction in communication links relative to the hard 
graph case and a communication density of $33\%$ compared to the $67\%$ density of the physical topology.

Figure \ref{fig:results} summarizes the closed-loop performance of the proposed framework. As shown in Fig. \ref{fig:results}(a), all DG voltages converge to the reference value $V_{\rm ref} = 311$ V from a cold start (i.e., starting from zero initial conditions), confirming effective voltage regulation with zero steady-state error across all DGs. Figure \ref{fig:results}(b) shows that all DG frequencies track the nominal reference value of $60$ Hz, demonstrating successful frequency synchronization. As shown 
in Fig. \ref{fig:results}(c), all DGs converge to the same normalized active power level of $P_i/P_i^{\rm max} = 0.6$, achieving proportional active power sharing. Figure \ref{fig:results}(d) shows that the normalized reactive power levels also converge, confirming effective reactive power sharing. Notably, the performance under the soft graph constraint is identical to that under the hard graph constraint across all four metrics, demonstrating that the proposed co-design framework can reduce communication infrastructure cost without sacrificing closed-loop performance.

\begin{figure}[t]
    \centering
    \includegraphics[width=\columnwidth]{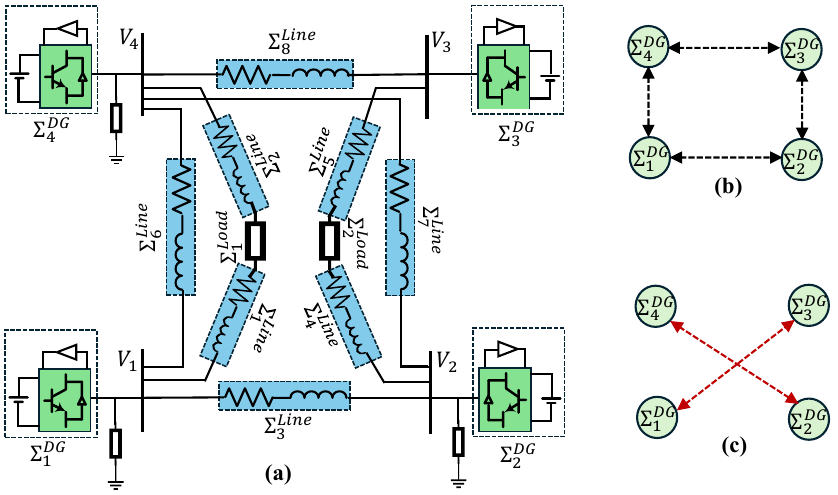}
     \vspace{-6mm}
   \caption{The AC MG topology with $N=4$ DGs, $L=8$ lines, and $M=2$ global loads: (a) physical, (b) hard graph communication matching physical adjacency, and (c) soft graph communication with reduced overhead.}
    \label{fig:topology}
\end{figure}

\begin{figure}[t]
    \centering
    \includegraphics[width=\columnwidth]{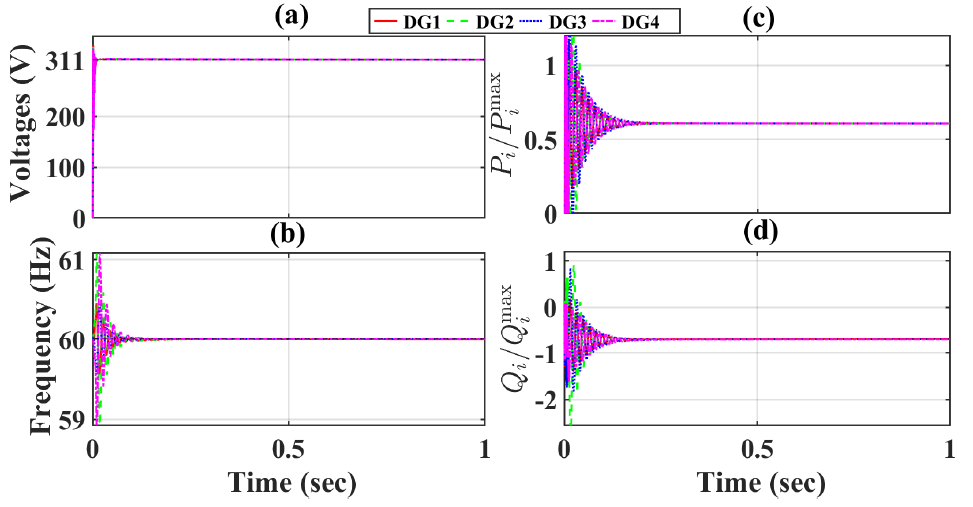}
    \vspace{-7mm}
    \caption{The closed-loop performance of the proposed dissipativity-based distributed controller in the AC MG with 4 DGs and 8 lines: (a) DG voltage 
    magnitudes, (b) DG frequencies, (c) normalized active power sharing, and (d) normalized reactive power sharing.}
    \vspace{-1mm}
    \label{fig:results}
\end{figure}

\section{Conclusion}\label{Conclusion}
This paper presented a dissipativity-based framework for co-designing distributed controllers and communication topologies in AC MGs to achieve robust guarantees for voltage regulation, frequency synchronization, and proportional power sharing. We first formulated the closed-loop AC MG as a standard networked system, while employing a hierarchical control architecture that combines local steady-state controllers for adjusting the operating point, local PI controllers for voltage regulation, and distributed consensus-type controllers for proportional power sharing. We next formulated the operating point design problem using equilibrium point analysis. Necessary and sufficient conditions for subsystem dissipativity analysis/design were then established while addressing complex subsystem dynamics. We then formulated the co-design problem as a convex optimization to simultaneously synthesize distributed controller gains and a sparse communication topology, while guaranteeing networked system dissipativity (robustness). Future work will focus on extending the developed co-design framework to account for communication delays and deploying moving target defense schemes for enhanced cybersecurity.

\bibliographystyle{IEEEtran}
\bibliography{References}

\end{document}